\pdfoutput=1
\documentclass[10pt, a4paper]{article}

\usepackage[utf8]{inputenc}
\usepackage[T1]{fontenc}
\usepackage{lmodern} 

\usepackage[margin=2cm]{geometry}

\usepackage[activate={true,nocompatibility},final,tracking=true, kerning=true,spacing=true,factor=1100,stretch=10,shrink=10]{microtype}
\SetTracking{encoding={*}, shape=sc}{40}
\usepackage{ellipsis}
\usepackage[abbreviations,british]{foreign}

\usepackage{etoolbox}

\usepackage{xcolor}
\usepackage{hyperref}

\definecolor{darkmidnightblue}{rgb}{0.0, 0.2, 0.4}
\definecolor{persianplum}{rgb}{0.44, 0.11, 0.11}

\hypersetup{
  colorlinks  = true,        
  citecolor   = persianplum, 
  urlcolor    = persianplum, 
  linkcolor   = persianplum
}

\usepackage{tikz}
\usepackage{todonotes}
\usepackage{xspace}

\usepackage{amssymb} 
\usepackage{amsmath}
\usepackage{amsthm}
\usepackage[bb=boondox]{mathalfa} 

\usepackage[capitalise, noabbrev]{cleveref}

\newtheorem{thm}{Theorem}[section]
\newtheorem{definition}[thm]{Definition}
\newtheorem{corollary}[thm]{Corollary}
\newtheorem{lemma}[thm]{Lemma}

 \makeatletter
 \def\desclabel#1#2{\begingroup
    \def\@currentlabel{#1}%
    #1\label{#2}\endgroup
 }
 \makeatother

\usepackage{caption}
\usepackage{subcaption}
\usepackage{float}

\usepackage[shortlabels]{enumitem} 
\usepackage{multicol}

\newcommand{\orcid}[1]{\href{https://orcid.org/#1}{\includegraphics[width=9pt]{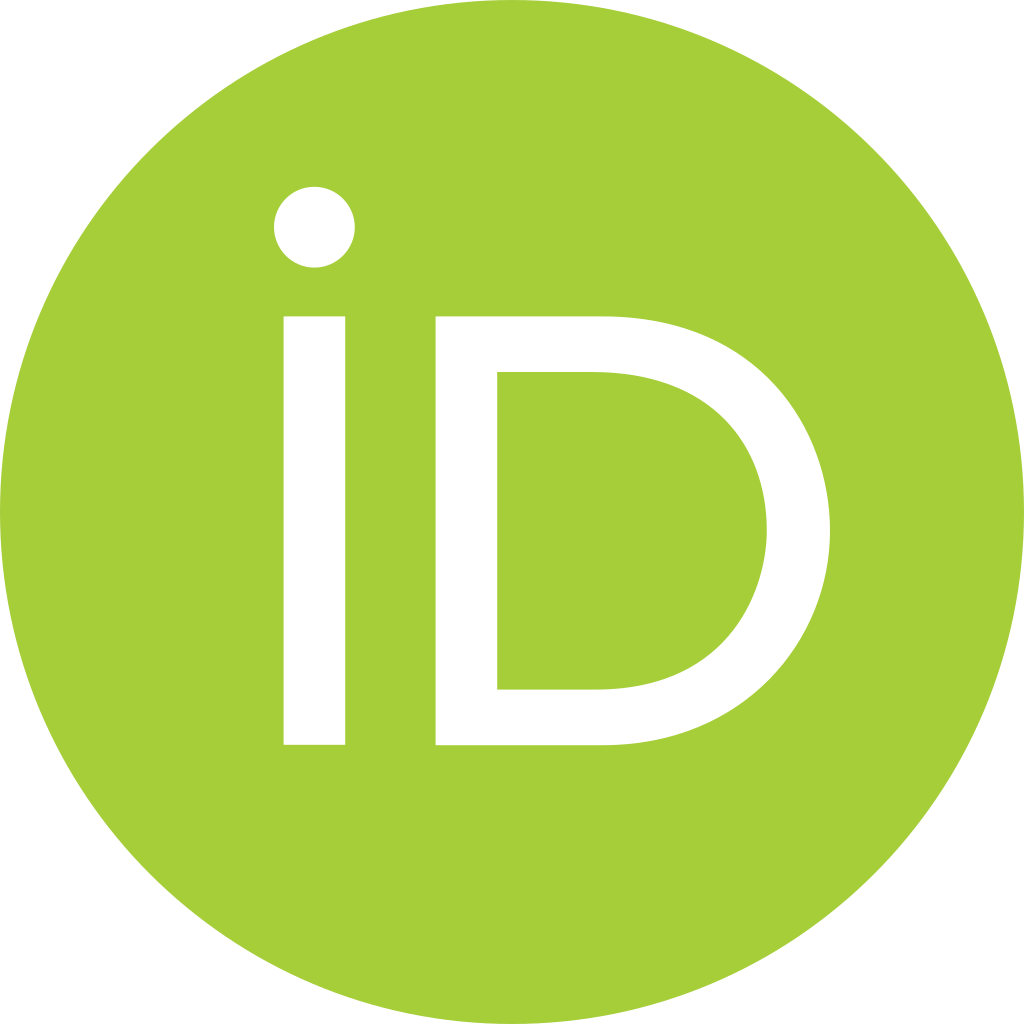}}}

 \makeatletter
 \def\desclabel#1#2{\begingroup
    \def\@currentlabel{#1}%
    #1\label{#2}\endgroup
 }
 \makeatother

\definecolor{ao(english)}{rgb}{0.0, 0.5, 0.0}

\newrobustcmd{\newnotion}[1]{\emph{#1}}

\newcommand{\dland}{\sqcap}
\newcommand{\dlor}{\sqcup}
\newcommand{\topconcept}{\top}
\newcommand{\botconcept}{\bot}
\newcommand{\dlsubseteq}{\sqsubseteq}
\newcommand{\Self}{\mathsf{Self}}
\newcommand{\bigdlor}{\bigsqcup}

\makeatletter
\providecommand{\bigdland}{%
  \mathop{%
    \mathpalette\@updown\bigsqcup
  }%
}
\newcommand*{\@updown}[2]{%
  \rotatebox[origin=c]{180}{$\m@th#1#2$}%
}
\makeatother

\newcommand{\DL}[1]{\ensuremath{\mathcal{#1}}}  
\newcommand{\ALC}{\DL{ALC}}                     
\newcommand{\ALCself}{\DL{ALC}_{\Self}}         
\newcommand{\SH}{\DL{SH}}                       
\newcommand{\ALCHbselfreg}{\DL{ALCH}\mathit{b}^{\Self}_{\mathsf{reg}}}
\newcommand{\Z}{\DL{Z}}                       
\newcommand{\FGF}{\DL{FGF}}   

\newcommand{\complexityclass}[1]{\textsc{#1}} 
\newcommand{\ExpTime}{\complexityclass{ExpTime}} 
\newcommand{\TwoExpTime}{\complexityclass{2ExpTime}} 
\newcommand{\coTwoExpTime}{\text{co}\complexityclass{2ExpTime}} 
\newcommand{\lang}[1]{\mathbf{#1}}  
\newcommand{\Ilang}{\lang{N_I}}     
\newcommand{\Rlang}{\lang{N_R}}     
\newcommand{\Clang}{\lang{N_C}}     
\newcommand{\Vlang}{\lang{N_V}}     

\newcommand{\query}[1]{\mathit{#1}}  
\newcommand{\queryq}{\query{q}}      
\newcommand{\match}[1]{#1}          
\newcommand{\modelsmatch}[1]{\models_{#1}} 
\newcommand{\matchpi}{\match{\pi}}  
\newcommand{\queryVar}[1]{\mathrm{Var}{(#1)}}   
\newcommand{\queryVarq}{\queryVar{\queryq}}     

\newcommand{\var}[1]{\mathit{#1}}   
\newcommand{\varx}{\var{x}}         
\newcommand{\vary}{\var{y}}         
\newcommand{\varz}{\var{z}}         

\newcommand{\role}[1]{\mathit{#1}}      
\newcommand{\roler}{\role{r}}           
\newcommand{\roles}{\role{s}}           
\newcommand{\rolet}{\role{t}}           

\newcommand{\concepts}{\lang{C}}            
\newcommand{\alcselfconcepts}{\concepts} 

\newcommand{\concept}[1]{\mathrm{#1}}       
\newcommand{\conceptA}{\concept{A}}         
\newcommand{\conceptB}{\concept{B}}         
\newcommand{\conceptC}{\concept{C}}         
\newcommand{\conceptD}{\concept{D}}         

\newcommand{\indv}[1]{\texttt{#1}}  
\newcommand{\indva}{\indv{a}}       
\newcommand{\indvb}{\indv{b}}       

\newcommand{\kb}[1]{\mathcal{#1}}   
\newcommand{\kbK}{\kb{K}}           
\newcommand{\abox}[1]{\mathcal{#1}} 
\newcommand{\aboxA}{\abox{A}}       
\newcommand{\tbox}[1]{\mathcal{#1}} 
\newcommand{\tboxT}{\tbox{T}}       

\newcommand{\inter}[1]{\mathcal{#1}}    
\newcommand{\interC}{\inter{C}}         
\newcommand{\interE}{\inter{E}}         
\newcommand{\interI}{\inter{I}}         
\newcommand{\interJ}{\inter{J}}         
\newcommand{\interU}{\inter{U}}         
\newcommand{\interQ}{\inter{Q}}         

\newcommand{\DeltaInter}[1]{\Delta^{#1}}                
\newcommand{\DeltaC}{\DeltaInter{\interC}}              
\newcommand{\DeltaE}{\DeltaInter{\interE}}              
\newcommand{\DeltaI}{\DeltaInter{\interI}}              
\newcommand{\DeltaJ}{\DeltaInter{\interJ}}              
\newcommand{\DeltaU}{\DeltaInter{\interU}}              
\newcommand{\DeltaQ}{\DeltaInter{\interQ}}              

\newcommand{\cdotInter}[1]{\cdot^{#1}}          
\newcommand{\cdotI}{\cdotInter{\interI}}        
\newcommand{\cdotQ}{\cdotInter{\interQ}}        

\newcommand{\domelem}[1]{\mathrm{#1}}                           
\newcommand{\domelemc}{\domelem{c}}                             
\newcommand{\domelemd}{\domelem{d}}                             
\newcommand{\domeleme}{\domelem{e}}                             

\newcommand{\deff}{:=}

\newcommand{\set}[1]{\{#1\}}

\renewcommand{\restriction}{\mathord{\upharpoonright}}
\newcommand{\restr}[2]{#1\restriction_{#2}} 

\newcommand{\atmM}{\mathcal{M}}
\newcommand{\nofM}{\mathtt{N}}
\newcommand{\statesQ}{\mathtt{Q}}
\newcommand{\statesQexists}{\statesQ_{\exists}}
\newcommand{\statesQforall}{\statesQ_{\forall}}
\newcommand{\states}{\mathtt{s}}
\newcommand{\initialstate}{\states_{I}}
\newcommand{\acceptingstate}{\states_{A}}
\newcommand{\rejectingstate}{\states_{R}}
\newcommand{\transreldelta}{\mathtt{T}} 
\newcommand{\transitiont}{\mathtt{t}}
\newcommand{\letterzero}{\scalebox{1.2}[0.9]{$\mathtt{0}$}}
\newcommand{\letterone}{\scalebox{1.2}[0.9]{$\mathtt{1}$}}
\newcommand{\lettera}{\mathtt{a}}
\newcommand{\letterb}{\mathtt{b}}
\newcommand{\tapeword}[1]{\mathtt{#1}}
\newcommand{\tapewordw}{\tapeword{w}}

\newcommand{\tapewordv}{\tapeword{v}}

\newcommand{\homo}[1]{\mathfrak{#1}}    
\newcommand{\homof}{\homo{f}}           
\newcommand{\homog}{\homo{g}}           
\newcommand{\homoh}{\homo{h}}           

\newcommand{\rolesunit}{\mathbf{R}_{\textit{unit}}}
\newcommand{\addressword}[1]{#1}
\newcommand{\addresswordu}{\addressword{u}}
\newcommand{\addresswordv}{\addressword{v}}
\newcommand{\addresswordw}{\addressword{w}}
\newcommand{\conceptsunit}{\mathbf{C}_{\textit{unit}}}
\newcommand{\roleell}{\role{\ell}}
\newcommand{\rolenext}{\role{next}}
\newcommand{\conceptLvl}{\concept{Lvl}}
\newcommand{\conceptL}{\concept{L}}
\newcommand{\conceptR}{\concept{R}}
\newcommand{\conceptadr}[1]{\concept{Ad}_{#1}}
\newcommand{\kbKunit}[1]{\kbK_{\textit{unit}}^{#1}}
\newcommand{\qrl}{\queryq_{\textit{rl}}} 

\newcommand{\conceptsconf}{\concepts_{\textit{conf}}}
\newcommand{\conceptzero}{\mathbb{0}}
\newcommand{\conceptone}{\mathbb{1}}
\newcommand{\conceptLetter}[1]{\concept{Let}_{#1}}
\newcommand{\currstate}[1]{\concept{St}_{#1}}
\newcommand{\currstates}{\currstate{\states}}
\newcommand{\HeadPos}[1]{\concept{HdPos}_{#1}}
\newcommand{\HeadHere}{\concept{HdHere}}
\newcommand{\NoHeadHere}{\concept{NoHdHere}}
\newcommand{\HeadLetter}[1]{\concept{HdLet}_{#1}}
\newcommand{\kbKconf}{\kbK_{\textit{conf}}}

\newcommand{\conceptsenrct}{\concepts_{\textit{enr}}}
\newcommand{\conceptsptrans}{\concepts_{\textit{ptr}}}
\newcommand{\kbKenr}{\kbK_{\textit{enr}}}
\newcommand{\InitialConf}{\concept{Init}}
\newcommand{\prevtrans}[1]{\concept{PTrns}_{#1}}
\newcommand{\PHeadPos}[1]{\concept{PHdPos}_{#1}}
\newcommand{\PHeadLetter}[1]{\concept{PHdLet}_{#1}}
\newcommand{\PrevHeadHere}{\concept{PHdHere}}
\newcommand{\NoPrevHeadHere}{\concept{NoPHdHere}}
\newcommand{\PrevHeadAbove}{\concept{PHdAbv}}
\newcommand{\NoPrevHeadAbove}{\concept{NoPHdAbv}}

\newcommand{\quasirunR}{\mathfrak{R}}
\newcommand{\treeT}{\mathfrak{T}}
\newcommand{\frakw}{\mathfrak{w}}

\newcommand{\frakb}{\mathfrak{b}}
\newcommand{\frakzero}{\mathfrak{0}}
\newcommand{\frakone}{\mathfrak{1}}

\newcommand{\qmain}{\queryq_{\textit{main}}}
\newcommand{\qaddr}[1]{\queryq_{\textit{adr}}^{#1}}
\newcommand{\qtopdown}{\queryq_{\downarrow}}

\newcommand{\qLtopdowni}[1]{\queryq_{\roleell{#1}\downarrow}}
\newcommand{\qRtopdowni}[1]{\queryq_{\roler{#1}\downarrow}}
\newcommand{\qithbit}[2]{\queryq^{#1\text{-th bit}}_{=#2}}

\title{The Price of Selfishness:\\ Conjunctive Query Entailment for $\ALCself$ is \mbox{\sc{2ExpTime}}-hard}

\author{Bartosz Bednarczyk$^{1,2}$ \orcid{0000-0002-8267-7554} and Sebastian Rudolph$^{1}$ \orcid{0000-0002-1609-2080} }
\date{$^{1}$Computational Logic Group, Technische Universit{\"a}t Dresden, Germany\\
$^{2}$Institute of Computer Science, University of Wroc{\l}aw, Poland\\
\texttt{ $\{$bartosz.bednarczyk, sebastian.rudolph$\}$@tu-dresden.de }}

\begin{document}
\maketitle

\begin{abstract}
In logic-based knowledge representation, query answering has essentially replaced mere satisfiability checking as the inferencing problem of primary interest.
For knowledge bases in the basic description logic $\ALC$, the computational complexity of conjunctive query (CQ) answering is well known to be \ExpTime-complete and hence not harder than satisfiability.
This does not change when the logic is extended by certain features (such as counting or role hierarchies), whereas adding others (inverses, nominals or transitivity together with role-hierarchies) turns CQ answering exponentially harder.

We contribute to this line of results by showing the surprising fact that even extending $\ALC$ by just the $\Self$ operator -- which proved innocuous in many other contexts -- increases the complexity of CQ entailment to $\TwoExpTime$.
As common for this type of problem, our proof establishes a reduction from alternating Turing machines running in exponential space, but several novel ideas and encoding tricks are required to make the approach work in that specific, restricted setting.  
\end{abstract}

\section{Introduction}\label{sec:intro}
  Formal ontologies are of significant importance in artificial intelligence, playing a central role in the Semantic Web, ontology-based information integration, or peer-to-peer data management.
  In such scenarios, an especially prominent role is played by \emph{description logics} (DLs)~\cite{dlbook} -- a robust family of logical formalisms used to describe ontologies and serving as the logical underpinning of contemporary standardised ontology languages.
  To put knowledge bases to full use as core part of intelligent information systems, much attention is being devoted to the area of ontology-based data-access, with  \emph{conjunctive queries} (CQs) being employed as a fundamental querying formalism~\cite{OrtizS12}.

  In recent years, it has become apparent that various modelling features of DLs affect the complexity of CQ answering in a rather strong sense.
  Let us focus on the most popular DL,~$\ALC$.
  It was first shown in~\cite[Thm. 2]{Lutz08} that CQ entailment is exponentially harder than the consistency problem for $\ALC$ extended with inverse roles $(\mathcal{I})$. 
  Shortly after, a combination of transitivity and role-hierarchies $(\mathcal{S}\mathcal{H})$ was shown as a culprit of higher worst-case complexity of reasoning~\cite[Thm. 1]{EiterLOS09}.
  Finally, also nominals $(\mathcal{O})$ turned out to be problematic~\cite[Thm. 9]{NgoOS16}.
  Nevertheless, there are also more benign DL constructs regarding the complexity of CQ entailment.
  Examples are counting~($\mathcal{Q}$)~\cite[Thm. 4]{Lutz08} (the complexity stays the same even for expressive arithmetical constraints~\cite[Thm. 21]{BaaderBR20}), role-hierarchies alone $(\mathcal{H})$~\cite[Cor.~3]{EiterOS12}, or even a tamed use of higher-arity relations~\cite[Thm. 20]{Bednarczyk21}.

  \paragraph{Our results.}
  We study CQ entailment in $\ALCself$, an extension of $\ALC$ with the $\Self$ operator, \ie a modelling feature that allows us to specify the situation when an element is related to it\emph{self} by a binary relationship. Among other things, this allows us to formalise the concept of a ``narcissist'':
  \[
      \concept{Narcissist} \dlsubseteq \exists{\role{loves}}.\Self,
  \]
  or to express that no person is their own parent:
  \[
      \concept{Person} \dlsubseteq \neg \exists{\role{hasParent}}.\Self.
  \]
  The $\Self$ operator is supported by the OWL 2 Web Ontology Language and the DL $\DL{SROIQ}$~\cite{HorrocksKS06}.
  Due to the simplicity of the $\Self$ operator (it only refers to one element), it is easy to accommodate for automata techniques~\cite{CalvaneseEO09} or consequence-based methods~\cite{OrtizRS10} and thus, so far, there has been no real indication that the added expressivity provided by $\Self$ may change anything, complexity-wise. Arguably, this impression is further corroborated by the observation that $\Self$ features in two profiles of OWL 2 (the EL and the RL profile), again without harming tractability \cite{KRH2008}.

  In this work, however, we show a rather counter-intuitive result, namely that CQ entailment for $\ALCself$ is exponentially harder than for $\ALC$.
  Hence, it places the seemingly innocuous $\Self$ operator among the ``malign'' modelling features, like $(\mathcal{I}), (\mathcal{SH})$, or $(\mathcal{O})$. 
  Moreover, this establishes $\TwoExpTime$-hardness of query entailment for the $\Z$ family (a.k.a. $\ALCHbselfreg$) of DLs~\cite{CalvaneseEO09}, which until now remained open\footnote{We stress that $\TwoExpTime$ hardness of CQ entailment over $\Z$ ontologies \emph{does not} follow from $\TwoExpTime$-hardness of the same problem for $\SH$ since we cannot define in $\Z$ that a given role is transitive nor that it is a transitive closure of another role (to simulate transitivity). Actually a modification of techniques by Lutz~\cite[Section~5]{Lutz08} yields $\ExpTime$ upper bound for querying of $\Z$ without~$\Self$~\cite{Bednarczyk21Spoiler}, which we are going to publish in the journal version of this paper.} as well as the $\TwoExpTime$-hardness of querying ontologies formulated in the forward guarded fragment~\cite{Bednarczyk21} with~equality.\footnote{Self-loops can be expressed in $\FGF_=$ via $\forall{\varx_1} \concept{self_r}(\varx_1) \to \exists{\varx_2} \roler(\varx_1, \varx_2) \land \varx_1 {=} \varx_2 \; \land \; \forall{\varx_1} \forall{\varx_2} \roler(\varx_1, \varx_2) \to (\varx_1 {=} \varx_2 \to \concept{self_r}(\varx_1))$.}

  Our proof goes via encoding of computational trees of alternating Turing machines working in exponential space and follows a general hardness-proof-scheme by Lutz~\cite[Section 4]{Lutz08}. 
  However, to adjust the schema to $\ALCself$, novel ideas are required: the ability to speak about self-loops is exploited to produce a single query that traverses trees in a root-to-leaf manner and to simulate disjunction inside CQs, useful to express that certain paths are repeated inside the tree.
  The axiomatisation used in this paper is not optimised for succinctness. In fact, many of our axioms are dispensable, but they provide a sharper understanding of the described structures and significantly the~proofs.

\section{Preliminaries}\label{sec:prelim}
  We recall the basics on description logics (DLs)~\cite{dlbook} and query answering~\cite{OrtizS12}.

  \paragraph*{DLs.}\label{subsec:prelim-dls}
  We fix countably-infinite pairwise disjoint sets of \emph{individual names} \(\Ilang \), \emph{concept names} \(\Clang \), and \emph{role names}~\(\Rlang \) and introduce the description logic \( \ALCself \).
  Starting from \(\Clang \) and \(\Rlang \), the set \(\alcselfconcepts \) of \( \ALCself \) \emph{concepts} is built using the following concept constructors: \emph{negation} \((\neg \conceptC) \), \emph{conjunction} \((\conceptC \dland \conceptD) \), \emph{existential restriction} (\(\exists{\roler}.\conceptC \)), the \emph{top concept} (\(\topconcept \)), and \emph{Self} concepts~\( (\exists{\roler}.\Self) \) , with the grammar:
  \begin{equation*} \label{eq:alc-grammar}
  \conceptC, \conceptD \; ::= \; \topconcept \; \mid \; \conceptA \; \mid \; \neg \conceptC \; \mid \; \conceptC \dland \conceptD \; \mid \; \exists{\roler}.\conceptC \; \mid \; \exists{\roler}.\Self,
  \end{equation*}
  where \(\conceptC,\conceptD \in \alcselfconcepts \), \(\conceptA \in \Clang \), and \(\roler \in \Rlang \). 
  We often employ disjunction \(\conceptC \dlor \conceptD \deff \neg (\neg \conceptC  \dland \neg \conceptD) \), universal restrictions \(\forall{\roler}.\conceptC \deff \neg \exists{\roler}.\neg\conceptC \), bottom \(\botconcept \deff \neg \topconcept \), and the less commonly used ``inline-implication'' $\conceptC \to \conceptD \deff \neg \conceptC \dlor \conceptD$.

  \emph{Assertions} are of the form \(\conceptC(\indva) \) or \(\roler(\indva,\indvb) \) for \(\indva,\indvb \in \Ilang \), \(\conceptC \in \alcselfconcepts \), and \(\roler \in \Rlang \).
  A \emph{general concept inclusion} (GCI) has the form \(\conceptC \dlsubseteq \conceptD \) for concepts \(\conceptC, \conceptD \in \alcselfconcepts \). 
  We use $\conceptC \equiv \conceptD$ as a shorthand for the two GCIs \(\conceptC \dlsubseteq \conceptD \) and \(\conceptD \dlsubseteq \conceptC \).
  A~\emph{knowledge base} (KB) \(\kbK=(\aboxA, \tboxT) \) is composed of a finite non-empty set \(\aboxA \) (\emph{ABox}) of assertions and a finite non-empty set \(\tboxT \) (\emph{TBox}) of GCIs. 
  We call the elements of \(\aboxA \cup \tboxT \) \emph{axioms}. 

  \begin{table}[!htb]
    \begin{minipage}{.55\linewidth}
      \vspace{0.045cm}
      \caption{Concepts and roles in \(\ALCself \).\label{tab:ALCself}}
      \centering
          \begin{tabular}{@{}l@{\ \ \ }c@{\ \ \ }l@{}}
              \hline\\[-2ex]
              Name & Syntax & Semantics \\ \hline \\[-2ex]
              top concept & \(\topconcept \) & \(\DeltaI  \) \\
              concept name & \(\conceptA \) & \(\conceptA^\interI \subseteq \DeltaI  \) \\ 
              role name & \(\roler \) & \(\roler^\interI \subseteq \DeltaI {\times} \DeltaI \) \\ 
              conc.\ negation & \(\neg\conceptC \)& \(\DeltaI \setminus \conceptC^{\interI} \) \\  
              conc.\ intersection & \(\conceptC \dland \conceptD \)& \(\conceptC^{\interI}\cap \conceptD^{\interI} \) \\  
              exist.\ restriction & \(\exists{\roler}.\conceptC \) & 
              \(\set{ \domelemd \mid \exists{\domeleme}.(\domelemd,\domeleme)\in \roler^{\interI} \land \domeleme\in \conceptC^{\interI}} \)\\
              Self concept & \(\exists{\roler}.\Self \) & 
              \(\set{ \domelemd \mid (\domelemd, \domelemd) \in \roler^{\interI} }\)   
              \\\hline
          \end{tabular}
    \end{minipage}%
    \begin{minipage}{.45\linewidth}
      \caption{Axioms in \( \ALCself \).\label{tab:axm}} 
      \centering
          \begin{tabular}{ l l }
              \hline\\[-2ex]
              Axiom \(\alpha \) & \(\interI \models\alpha \), if \\ \hline \\[-2ex] 
              \(\conceptC \dlsubseteq \conceptD \) & \(\conceptC^{\interI} \subseteq \conceptD^{\interI}  \)\hspace{5ex} \mbox{TBox}~\(\tboxT \) \\ \hline \\[-2ex]
              \(\conceptC(\indva) \) & \(\indva^{\interI} \in \conceptC^{\interI}  \)\hfill \mbox{ABox}\(~\aboxA \) \\
              \(\roler(\indva,\indvb) \) & \((\indva^{\interI}, \indvb^{\interI} )\in \roler^{\interI}  \)             
              \\\hline
          \end{tabular}
    \end{minipage} 
\end{table}
  The semantics of \(\ALCself \) is defined via \emph{interpretations} \(\interI = (\DeltaI, \cdotI) \) composed of a non-empty set \(\DeltaI \) called the \emph{domain of \(\interI \)}, and an \emph{interpretation function} \(\cdotI \) mapping individual names to elements of \(\DeltaI \), concept names to subsets of \(\DeltaI \), and role names to subsets of \(\DeltaI \times \DeltaI \). 
  This mapping is extended to concepts (see~\cref{tab:ALCself}) and finally used to define \emph{satisfaction} of assertions and GCIs (see~\cref{tab:axm}). 
  We say that an interpretation \(\interI \) \emph{satisfies} a KB \(\kbK=(\aboxA,\tboxT) \) (or \(\interI \) is a \emph{model} of \(\kbK \), written: \(\interI \models \kbK \)) if it satisfies all axioms of~\(\aboxA\cup\tboxT \). 
  A KB is \emph{consistent} (or \emph{satisfiable}) if it has a model, and \emph{inconsistent} (or \emph{unsatisfiable}) otherwise. 

A \emph{homomorphism} $\homoh : \interI \to \interJ$ is a concept-name and role-name-preserving function that maps every element of $\DeltaI$ to some element from $\DeltaJ$, \ie we have that $\domelemd \in \conceptA^{\interI}$ implies that $\homoh(\domelemd) \in \conceptA^{\interJ}$, and $(\domelemd, \domeleme) \in \roler^{\interI}$ implies $(\homoh(\domelemd), \homoh(\domeleme)) \in \roler^{\interJ}$ for all role names $\roler \; \in \Rlang$, concept names $\conceptA \; \in \Clang$, and $\domelemd, \domeleme \in \DeltaI$.

  \paragraph*{Queries.}\label{subsec:prelim-queries}
  \emph{Conjunctive queries} (CQs) are conjunctions of \emph{atoms} of the form \(\roler(\varx, \vary) \) or \(\conceptA(\varz) \), where \(\roler \) is a role name, \(\conceptA \) is a concept name, and \(\varx, \vary, \varz \) are variables from a countably infinite set \(\Vlang \).
  We denote with $|\queryq|$ the number of its atoms, and with $\queryVar{\queryq}$ the set of all variables that appear in $\queryq$.
  Let $\interI$ be an interpretation, $\queryq$ a CQ and $\matchpi: \queryVar{\queryq}\to\DeltaI$ be a variable assignment.
  We write \(\interI \modelsmatch{\matchpi} \roler(\varx,\vary) \) if~\((\matchpi(\varx),\matchpi(\vary))\in \roler^\interI \), and~\(\interI \modelsmatch{\matchpi} \conceptA(\varz) \) if \(\matchpi(\varz) \in \conceptA^\interI \). 
  We say that~\(\matchpi \) is a \emph{match} for \(\interI \) and \(\queryq \) if \(\interI \modelsmatch{\matchpi} \alpha \) holds for every atom $\alpha \in \queryq$, and that~\(\interI \) \emph{satisfies} \(\queryq \) (denoted with: \(\interI \models \queryq \)) whenever \(\interI \modelsmatch{\matchpi} \queryq \) for some match \(\matchpi \). 
  The definitions are lifted to KBs: \(\queryq \) is \emph{entailed} by a KB \(\kbK \) (written:~\(\kbK \models \queryq \)) if every model \(\interI \) of~\( \kbK \) satisfies~\(\queryq \). 
  We stress here that satisfaction of conjunctive queries is preserved by homomorphisms, \ie if $\interI \models \queryq$ and there is a homomorphism from $\homoh: \interI \to \interJ$ then $\interJ \models \queryq$. 
  When \(\interI \models \kbK \) but \(\interI \not\models \queryq \), we call \(\interI \) a \emph{countermodel} for \(\kbK \) and \(\queryq \). 
  The \emph{query entailment problem} asks whether $\kbK \models \queryq$ holds for an input KB $\kbK$ and a CQ $\queryq$.

  Whenever convenient, we employ the \emph{path syntax} of CQs to write queries in a concise way.
  By a \emph{path expression} we mean an expression of the form
  \[
  (\conceptA_0?;\roler_1; \conceptA_1?; \roler_2; \conceptA_2?; \ldots; \conceptA_{n{-}1}?; \roler_n; \conceptA_{n}?)(\varx_0, \varx_{n})
  \]
  with all $\roler_i \in \Rlang$ and $\conceptA_i \in \Clang$, serving as a shorthand for
  \[
  \bigwedge_{i=0}^{n} \conceptA_i(\varx_i) \land \bigwedge_{i=1}^{n} \roler_i(\varx_{i{-}1}, \varx_i).
  \]
  Whenever $\conceptA_i$ happens to be $\top$, it will be removed from the expression; this does not create ambiguities. Note that path CQs are just syntactic sugar and should not be mistaken \eg with regular path~queries.

  \subsection{Alternating Turing Machines}\label{subsec:prelim-atms}
  We next fix the notation of alternating Turing machines over a binary alphabet $\set{\letterzero,\letterone}$ working in exponential space (simply: ATMs).
  An ATM is defined as a tuple $\atmM = (\nofM, \statesQ, \statesQexists, \initialstate, \acceptingstate, \rejectingstate, \transreldelta )$,
  where \( \statesQ \) is a finite set of \emph{states} (usually denoted with \( \states \));
  \( \statesQexists \subseteq \statesQ \) is a set of \emph{existential} states;
  \( \initialstate, \acceptingstate, \rejectingstate \in \statesQ \) are, respectively, pairwise different \emph{initial}, \emph{accepting}, and \emph{rejecting} states; we assume that $\initialstate \in \statesQforall$
  \( \transreldelta \subseteq (\statesQ \times \set{\letterzero,\letterone}) \times (\set{\letterzero,\letterone} \times \statesQ \times \set{{-}1, {+}1}) \) is the \emph{transition relation};
  and the natural number $\nofM$ (encoded in unary) is a parameter governing the size of the working tape.
  We call the states from \( \statesQforall{} \deff \statesQ \setminus \statesQexists \) \emph{universal}. 
  The size of $\atmM$, denoted with $|\atmM|$, is~$\nofM + |\statesQ| + |\statesQexists| + 3 + |\transreldelta|$.

  \noindent A \emph{configuration} of $\atmM$ is a word \( \tapewordw \states \tapewordw' \in \set{\letterzero,\letterone}^* \statesQ \set{\letterzero,\letterone}^* \) with $|\tapewordw\tapewordw'| = 2^{\nofM}$.
  We call \( \tapewordw \states \tapewordw' \) (i) existential (resp. universal) if \( \states \) is existential (resp. universal), (ii) final if \( \states \) is either $\acceptingstate$ or $\rejectingstate$ (iii) non-final if it is not final (iv) accepting if \( \states = \acceptingstate \). 
  Successor configurations are defined in terms of the transition relation~$\transreldelta$.
  For $\mathtt{a},\mathtt{b},\mathtt{c},\mathtt{d} \in \set{\letterzero,\letterone}$ and $\mathtt{v},\mathtt{v}',\mathtt{w},\mathtt{w}' \in\set{\letterzero,\letterone}^*$ with $|\tapewordv| = |\tapewordw|$, we let $\tapewordw\mathtt{b}\states'\tapewordw'$ be a \emph{quasi-successor} configuration of $\tapewordv\states\mathtt{a}\tapewordv'$ whenever $(\states,\mathtt{a},\mathtt{b},\states',+1) \in \transreldelta$, and we let $\tapewordw\states'\mathtt{d}\mathtt{b}\tapewordw'$ be a quasi-successor configuration of $\tapewordv\mathtt{c}\states\mathtt{a}\tapewordv'$ whenever $(\states,\mathtt{a},\mathtt{b},\states',-1) \in \transreldelta$. 
  If additionally we meet the requirement $\tapewordw = \tapewordv$, $\tapewordw' = \tapewordv'$, and $\mathtt{c}=\mathtt{d}$ we speak of \emph{successor} configurations.\footnote{In words, this corresponds to the common definition of successor configurations, while for quasi-successor configurations, untouched tape cells may change arbitrarily during the transition.}

  Without loss of generality, we make the following additional assumptions about $\atmM$: 
  First, for each non-final (\ie{} non-accepting and non-rejecting) state \( \states \) and every letter $\mathtt{a} \in \{ \letterzero, \letterone \}$ the set \( \transreldelta(\states, \mathtt{a}) \deff \set{ (\states, \mathtt{a}, \mathtt{b}, \states', d) \in \transreldelta}  \) contains exactly two elements, denoted $\transreldelta_1(\states, \mathtt{a})$ and $\transreldelta_2(\states, \mathtt{a})$. Hence, every configuration has exactly two successor configurations.
  Second, for any $(\states, \mathtt{a}, \mathtt{b}, \states', d) \in \transreldelta$, if  \( \states \) is existential then \(\states'\) is universal and vice versa.
  Third, the machine reaches a final state no later than after $2^{2^{\nofM}}$ steps (for configuration sequences).
  Fourth and last, $\atmM$ never attempts to move left (resp. right) on the left-most (resp. right-most) tape cell.

  A \emph{run} of \( \atmM \) is a finite tree, with nodes labelled by configurations of \( \atmM \), that satisfies all the conditions~below:
  \begin{itemize}[itemsep=0pt]
      \item  the root is labelled with the initial configuration \(\initialstate \letterzero^{2^{\nofM}}\),
      \item each node labelled with a non-final existential configuration \( \tapewordw \states \tapewordw' \) has a single child node which is labelled with one of the successor configurations of \( \tapewordw \states \tapewordw' \),
      \item each node labelled with a non-final universal configuration \( \tapewordw \states \tapewordw' \) has two child nodes which are labelled with the two successor configurations (wrt. $\transreldelta_1$ and $\transreldelta_2$) of \( \tapewordw \states \tapewordw' \),
      \item no node labelled with a final configuration has successors. 
  \end{itemize}
  \emph{Quasi-runs} of \( \atmM \) are defined analogously by replacing the notions of successors with quasi-successors. Note that every run is also a quasi-run but not vice versa.
  
  \noindent An ATM \( \atmM \) is (quasi-)\emph{accepting} if it has an  \emph{accepting} (quasi)-run, \ie one whose all leaves are labelled by accepting configurations. By~\cite[Corollary 3.6]{ChandraS76} the problem of checking if a given ATM is accepting is~$\TwoExpTime$-hard.

\section{A High-Level Overview of the Encoding}\label{sec:overview}

Let $\atmM$ be an ATM.\@
The core contribution of our paper is to present a polynomial-time reduction that, given~$\atmM$, constructs a pair $(\kbK_{\atmM}, \queryq_{\atmM})$ --- composed of an $\ALCself$ knowledge base and a conjunctive query --- such that $\kbK_{\atmM} \not\models \queryq_{\atmM}$ iff $\atmM$ is accepting.
Intuitively, the models of $\kbK$ will encode accepting quasi-runs of $\atmM$, \ie trees in which every node is a meaningful configuration of $\atmM$, but the tape contents of consecutive configurations might not be in sync as they should.
The query $\queryq_{\atmM}$ will be responsible for detecting such errors.
Hence, the existence of a countermodel for $\kbK_{\atmM}$ and $\queryq_{\atmM}$ will coincide with the existence of an accepting run of $\atmM$.
The intended models of $\kbK_{\atmM}$ look as follows: 
\begin{figure}[H]
\centering
\vspace{-1em} 
\hspace{5em}  
\includegraphics[scale=0.14]{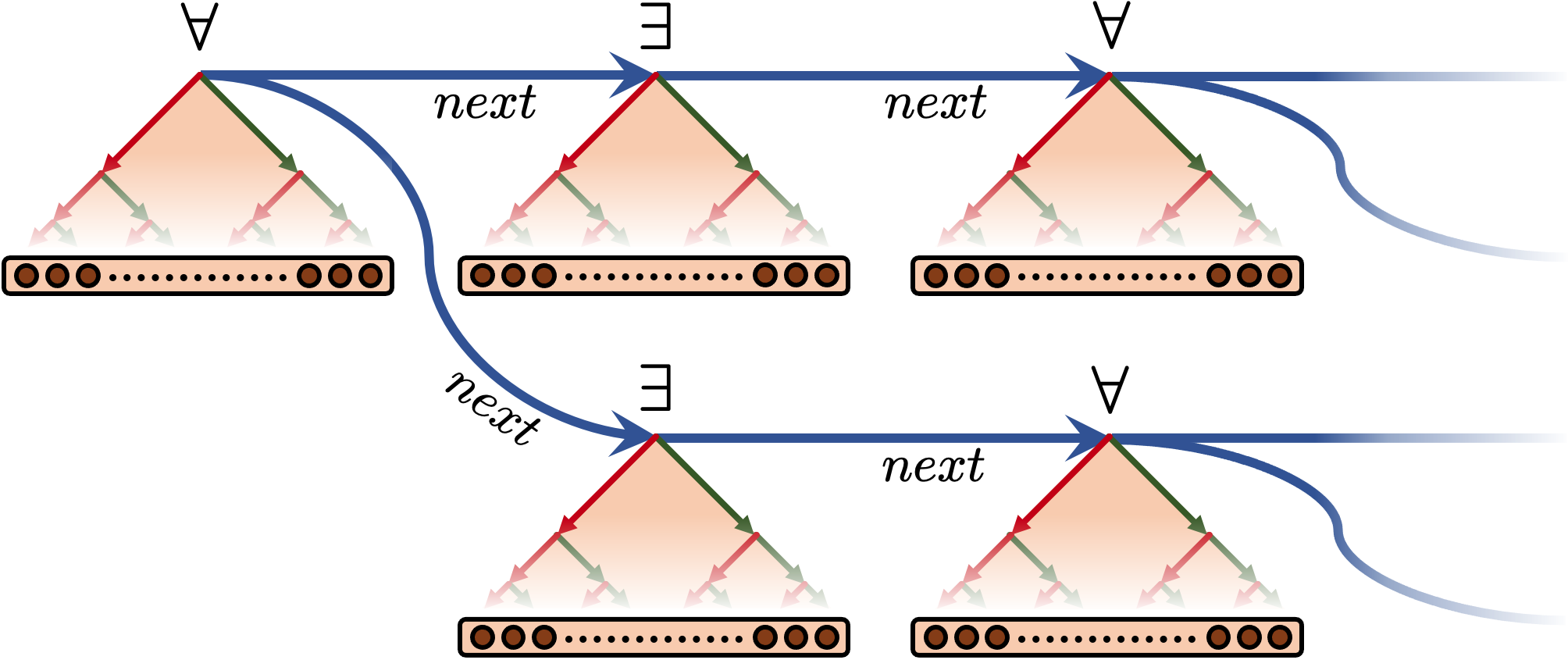}
\end{figure}

The depicted triangles are called the \emph{configuration trees} and encode configurations of $\atmM$. 
The information contained in these configuration trees is ``superimposed'' on identical \emph{configuration units}: full binary trees of height $\nofM{+}1$ decorated with many self-loops\footnote{The concrete purpose of the abundant presence of self-loops will only become clear later, starting from~\cref{lemma:query-from-root-to-leaf}.} that will provide the ``navigational infrastructure'' for the query $\queryq_{\atmM}$ to detect ``tape mismatches''. 
Every such tree has $2^{\nofM}$ nodes at its $\nofM$-th level and each of these nodes represents a single tape cell of a machine.
%
%
%
However, somehow unexpectedly, we do not just label them directly with concepts representing a letter from the alphabet. 
Instead, every node at the $\nofM$-th level also has two children labelled (from left to right) respectively with either $\conceptzero$ and $\conceptone$, or with $\conceptone$ and~$\conceptzero$. 
Whenever the left child is in $\conceptzero$ and the right child is in $\conceptone$, we think that their parent represents a cell filled with the letter $\letterzero$, while the converse situation encodes a cell filled with $\letterone$.
\vspace{-1em} 
\begin{figure}[H]
  \centering
\footnotesize
\begin{tabular}{c@{\hspace{5ex}}c}
  \includegraphics[width=0.1\columnwidth]{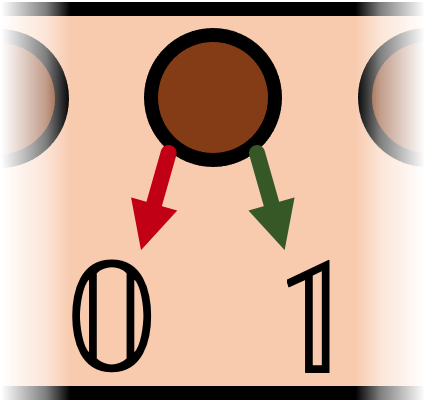} &
  \includegraphics[width=0.1\columnwidth]{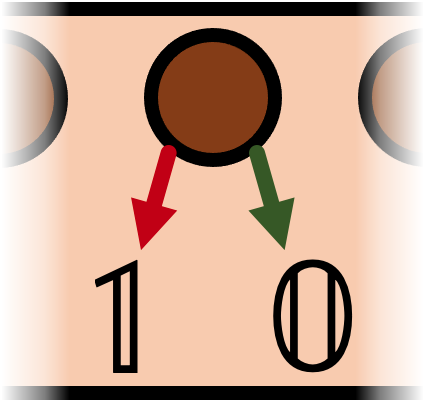}  \\
  encoding of tape cell & encoding of tape cell\\
  carrying symbol $\letterzero$ & carrying symbol $\letterone$ \\
\end{tabular}
\end{figure}
This encoding will be useful to avoid a seemingly required disjunction in the construction of $\queryq_{\atmM}$. 
Lastly, the roots of configuration units store all remaining necessary information required for encoding: the current state of $\atmM$, the previous and the current head position as well as the transition used to arrive at this node from the previous configuration.
Finally, the roots of configuration trees are interconnected by the role $\rolenext$ indicating that $(\domelem{r}, \domelem{r}') \in \rolenext^\interI$ holds iff the configuration represented by $\domelem{r}'$ is a quasi-successor of the configuration of $\domelem{r}$.

\section{Configuration Units}\label{sec:units}
In our encoding, a vital role is played by $n$-\emph{configuration units}, which will later form the backbone of configuration~trees. Roughly speaking, each $n$-configuration unit is a full binary tree of depth $n$, decorated with certain concepts, roles, and self-loops. 
We introduce configuration units by providing the formal definition, followed by a graphical depiction and an intuitive description.
In order to represent configuration units inside interpretations, we employ role names from $\rolesunit$ as well as concept names from $\conceptsunit$:
\begin{align*}
\rolesunit  := \{ \roleell_i, \roler_i, \rolenext \mid 1 \leq i \leq n\}, \qquad
\conceptsunit := \{ \conceptLvl_0, \conceptLvl_i, \conceptL, \conceptR, \conceptadr{i}^0, \conceptadr{i}^1 \mid 1 \leq i \leq n\}.
\end{align*}

\begin{definition}[configuration unit]\label{def:conf-units}
    Given a number $n$, an \emph{$n$-configuration unit} $\interU$ is an interpretation~$(\DeltaU,\cdot^\interU)$~with
    \begin{multicols}{2}
    \begin{itemize}\itemsep0em
    \item $\DeltaU = \{0,1\}^{\leq n} := \{ \addresswordw \in \{0,1\}^* \mid |\addresswordw| \leq n \}$,
    \item $(\conceptadr{i}^b)^{\interU} = \{ \addresswordw {\in} \DeltaU \mid |\addresswordw| \geq i \; \text{and its i-th letter is} \; b\}$,
    \item $\roleell_i^{\,\interU} = \{ (\addresswordw, \addresswordw0) \mid |\addresswordw| = i{-}1 \} \cup \{ (\addresswordw,\addresswordw) \mid \addresswordw \in \DeltaU \}$,
    \item $\conceptL^{\interU} \! \setminus \! \{ \varepsilon \} = \{ \addresswordw0 \in \DeltaU \}$,
    \end{itemize}
    \columnbreak
    \begin{itemize}\itemsep0em
    \item $\conceptLvl_i^{\interU} = \{ \addresswordw \in \DeltaU \mid |\addresswordw| = i\}$,
    \item $\rolenext^{\interU} = \{ (\addresswordw,\addresswordw) \mid |\addresswordw| = n \}$,
    \item $\roler_i^{\interU} = \{ (\addresswordw,\addresswordw1) \mid |\addresswordw| = i{-}1 \} \cup \{ (\addresswordw,\addresswordw) \mid \addresswordw \in \DeltaU \}$,
    \item $\conceptR^{\interU} = \DeltaU \setminus \conceptL^{\interU}$.
    \end{itemize}
    \end{multicols}
\end{definition}
\smallskip
\noindent The following drawing depicts a $2$-configuration unit.
\begin{figure}[H]
    \centering
    \includegraphics[scale=0.15]{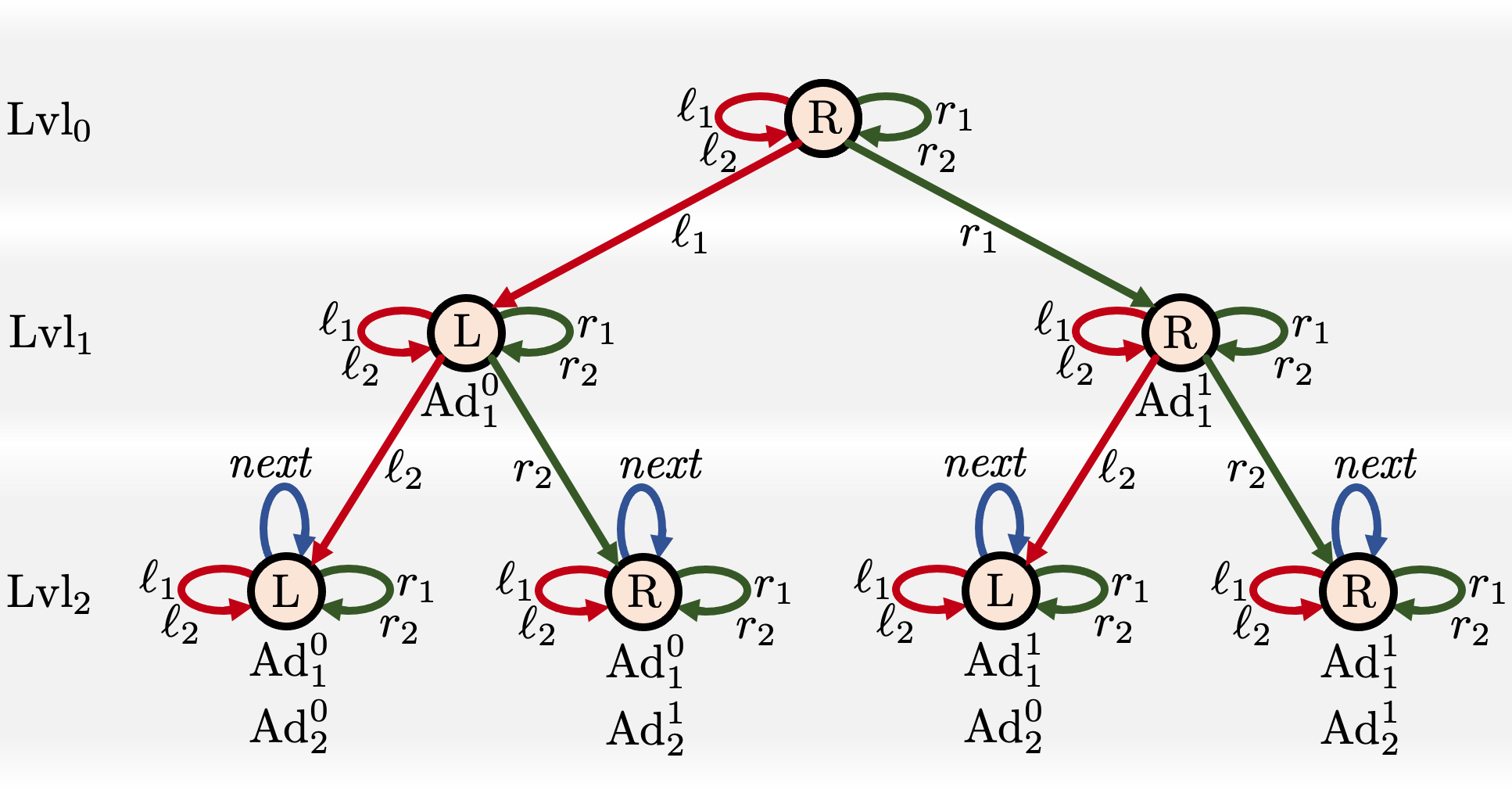}
\end{figure}

\noindent As one can see, the nodes in the tree are layered into levels according to their distance from the root. Nodes at the $i$-th level are members of the $\conceptLvl_i$ concept and their distance from the root is equal to $i$.
Next, each non-leaf node at the $i$-th level has two children, the left one and the right one (satisfying, respectively, the concepts $\conceptL$ and~$\conceptR$) and is connected to them via the role $\roleell_i$ and $\roler_i$, respectively. All nodes are equipped with $\roleell_i$- and $\roler_i$-self-loops and all leaves are additionally endowed with $\rolenext$-loops.
With all nodes inside the tree, we naturally associate their addresses, \ie their ``numbers'' when nodes from the $i$-th level are enumerated from left to right. In order to encode the address of a given node at the $i$-th level, we employ concepts $\conceptadr{1}^b, \conceptadr{2}^b, \ldots, \conceptadr{i}^b$ with ``values'' $b$ either $0$ or $1$, meaning that a node is in $\conceptadr{j}^b$ iff the $j$-th bit of its address is equal to~$b$. The most significant bit~is~$\conceptadr{1}^b$.\\

\noindent We next proceed with an axiomatisation of $n$-configuration units in $\ALCself$, obtained with the following GCIs.%
\footnote{As usual in such DL encodings, we cannot formalise such structures up to isomorphism, but the axiomatisation provided is sufficient in a sense made formally precise in the subsequent lemmas.}

\begin{enumerate}
    \item Each node is at exactly one level. 
    \begin{multicols}{2}
    \begin{description}\itemsep0em
        \item[\desclabel{(LvlCov)}{axiom:LvlCov}] $\topconcept \dlsubseteq \bigdlor_{i=0}^{n} \conceptLvl_i$
    \end{description}
    \columnbreak
    \begin{description}\itemsep0em
        \item[\desclabel{(LvlDisj[i,j])}{axiom:LvlDisjij}] $\conceptLvl_i \dland \conceptLvl_j \dlsubseteq \botconcept \; (\text{with} \; 0 \leq i < j \leq n)$
    \end{description}
    \end{multicols}

    \item All nodes carry self-loops for all role names from $\rolesunit$ except $\rolenext$ and all leaf nodes (and only they) carry a $\rolenext$-loop.
    \begin{multicols}{2}
    \begin{description}\itemsep0em
        \item[\desclabel{(all-loops-but-next)}{axiom:allloopsbutnext}] $\topconcept \dlsubseteq \bigdland_{\roles \in \rolesunit \setminus \{ \rolenext \}} \exists{\roles}.\Self \quad$
    \end{description}
    \columnbreak
    \begin{description}\itemsep0em
        \item[\desclabel{(leaves-next-loop)}{axiom:leavesnextloop}] $\conceptLvl_n \equiv \exists{\rolenext}.\Self$
    \end{description}
    \end{multicols}

    \item Every node is either a ``left'' node or a ``right'' node.
    \begin{multicols}{2}
    \begin{description}\itemsep0em
        \item[\desclabel{(LRCov)}{axiom:LRCov}] $\topconcept \dlsubseteq \conceptL \dlor \conceptR$
    \end{description}
    \columnbreak
    \begin{description}\itemsep0em
        \item[\desclabel{(LRDisj)}{axiom:LRDisj}] $\conceptL \dland \conceptR \dlsubseteq \botconcept$
    \end{description}
    \end{multicols}

    \item Each node at any level $0 \leq i < n$ has two successors (one left and one right).
    \begin{description}\itemsep0em
        \item[\desclabel{(LsuccLvl[i])}{axiom:LsuccLvli}] $\conceptLvl_{i} \, \dlsubseteq \exists{\roleell_{i{+}1}}.(\conceptLvl_{i{+}1}) \dland \forall{\roleell_{i{+}1}}.(\conceptLvl_{i{+}1} \to \conceptL)$
        \item[\desclabel{(RsuccLvl[i])}{axiom:RsuccLvli}] $\conceptLvl_{i} \dlsubseteq \exists{\roler_{i{+}1}}.(\conceptLvl_{i{+}1}) \dland \forall{\roler_{i{+}1}}.(\conceptLvl_{i{+}1} \to \conceptR)$
    \end{description}

    \item Address information for the nodes is created bit-wise and propagated down the tree.\\
    That is, once we are in the left (resp. right) node on the $i$-th level, this node and all nodes further below will have the $i$-th bit of their address set to $0$ (resp. $1$). Below we have $1 \leq i \leq n$, $b \in \{ 0,1 \}$ and $0 \leq j < i$.
    \begin{multicols}{2}
    \begin{description}\itemsep0em
        \item[\desclabel{(LBitZero[i])}{axiom:LBitZeroi}] $\conceptLvl_{i} \dland \conceptL \dlsubseteq \conceptadr{i}^0$
        \item[\desclabel{(AdDisj[i])}{axiom:AdDisji}] $\conceptadr{i}^0 \dland \conceptadr{i}^1 \dlsubseteq \botconcept$
    \end{description}
    \columnbreak
    \begin{description}\itemsep0em
        \item[\desclabel{(RBitOne[i])}{axiom:RBitOnei}] \;\;\;\;\;\;\;\;$\conceptLvl_{i} \dland \conceptR \dlsubseteq \conceptadr{i}^1$
        \item[\desclabel{(AdLvlDisj[i,j])}{axiom:AdLvlDisjij}] $\conceptadr{i}^b \dland \conceptLvl_{j} \dlsubseteq \botconcept $
    \end{description}
    \end{multicols}
    \begin{description}\itemsep0em
        \centering
        \item[\desclabel{(PropBit[i])}{axiom:propbiti}] $\conceptadr{i}^b \dlsubseteq \textstyle\bigdland_{j=1}^{n} \forall{\roleell_{j}}.\conceptadr{i}^b \dland \forall{\roler_{j}}.\conceptadr{i}^b$
    \end{description}
\end{enumerate}

This finishes the axiomatisation of $n$-configuration units.
Let $\kbKunit{n}$ denote the KB composed of all GCIs presented so far.
What remains to be done is to show that our axiomatisation is correct, in the sense of the following two lemmas.
Their proofs are routine, hence the reader may skip them at first reading.

\begin{lemma}\label{prop:n-conf-units-satisfies-KBKunit}
  Each $n$-configuration unit is a model of~$\kbKunit{n}$.
\end{lemma}
\begin{proof}
    Let $\interU$ be an $n$-configuration unit. We will proceed with all axioms $\alpha$ of $\kbKunit{n}$ and show that $\interU \models \alpha$.
    \begin{enumerate}\itemsep0em
        \item Note that $\DeltaU$ is a set of all binary words of length $\leq n$ and by definition we have that $\addresswordw \in \conceptLvl_i^{\interU}$ iff $|\addresswordw|=i$. Since every word $\addresswordw$ has a \emph{unique} length, it follows that $\interU \models \ref{axiom:LvlCov}$ and that $\interU \models \ref{axiom:LvlDisjij}$ for any indices $0 \leq j < i \leq n$.
        \item By definitions of $\conceptLvl_n^{\interU}$ and $\rolenext^{\interU}$, we immediately conclude $\interU \models$ \ref{axiom:leavesnextloop}.
              Moreover, for any role name $\roles$ from $\rolesunit \setminus \{ \rolenext \}$, we conclude $\interU \models \exists{\roles}.\Self$ from the fact that the set $\{ (\addresswordw, \addresswordw) \mid \addressword \in \DeltaU \}$ is explicitly stated as a part of $\roles^{\interU}$ in its definition.
            Thus, we infer $\interU \models$ \ref{axiom:allloopsbutnext}.
        \item The satisfaction of \ref{axiom:LRCov} and \ref{axiom:LRDisj} by $\interU$ is due to the equality $\conceptR^{\interU} = \DeltaU \setminus \conceptL^{\interU}$.
        \item Suppose $i < n$. We will show that $\interU \models$ \ref{axiom:LsuccLvli} (the satisfaction of \ref{axiom:RsuccLvli} is analogous).
        Hence, take any $\addresswordw \in \conceptLvl_i^{\interU}$. 
        Then $\addresswordw$ (by definition of $\roleell_{i{+}1}$) has exactly two $\roleell_{i{+}1}$-successors: $\addresswordw$ and $\addresswordw 0$. 
        Moreover, by definition of $\conceptLvl_{i{+}1}^{\interU}$ and $\conceptL^{\interU}$ we conclude that $\addresswordw 0 \in \conceptLvl_{i{+}1}^{\interU} \cap \conceptL^{\interU}$ and $\addresswordw \not\in \conceptLvl_{i{+}1}$. 
        Hence, $\interU \models$ \ref{axiom:LsuccLvli} holds, since $\addresswordw 0$ is the required (only) witness for the $\exists$- ($\forall$-) restriction.
        \item To see $\interU \models$ \ref{axiom:LBitZeroi}, it suffices to take any element $\addresswordw \in \conceptLvl_{i}^{\interU} \cap \conceptL^{\interU}$. 
        By the first inclusion we infer that $|\addresswordw|=i$ and by the second that the last letter of $\addresswordw$ is $0$. 
        Hence, we are done by the definition of $(\conceptadr{i}^b)^{\interU}$. 
        Similarly, we get $\interU \models$ \ref{axiom:RBitOnei}.
        The property $\interU \models$ \ref{axiom:AdDisji} is due to the fact that words from $\interU$ carry only one letter per position. 
        Next, to show $\interU \models \ref{axiom:AdLvlDisjij}$ for any $1 \leq j < i \leq n$ it suffices to see that, by definition, $\conceptLvl_j^{\interU}$ contains words of length $=j$ and $(\conceptadr{i}^b)^{\interU}$ contains words of length $\geq i$.
        Thus their intersection is empty, implying the satisfaction of \ref{axiom:AdLvlDisjij}.
        Finally, we need to prove $\interU \models \ref{axiom:propbiti}$ for $1 \leq i \leq n$. 
        To this end, note that $\addresswordw \in (\conceptadr{i}^b)^{\interU}$ is equivalent to saying that $|\addresswordw| \geq i$ and that the $i$-th letter of $\addresswordw$ is $b$.
        Now observe that, by definition, that for every $\roles \in \rolesunit \setminus \{ \rolenext \}$ we have that any $\roles$-successor of $\addresswordw$ can only be $\addresswordw$, $\addresswordw 0$, or $\addresswordw 1$ (if $i \neq 0$). 
        In any case, such a successor has length $\geq i$ and has its $i$-th letter equal to $b$. 
        Thus its membership in $(\conceptadr{i}^b)^{\interU}$ follows, finishing the proof. \qedhere
    \end{enumerate}
\end{proof}

The proof of the next lemma is by constructing an $n$-configuration unit by starting from $\domelemd$ and recursively traversing $\roleell_i$ and $\roler_i$ roles.
\begin{lemma}\label{lemma:n-conf-units-homomorphisms}
For any model~$\interI$ of $\kbKunit{n}$ and any $\domelemd \in \conceptLvl_0^{\interI}$ there is an $n$-configuration unit \;$\interU$ and a homomorphism $\homoh$ from $\interU$ into $\interI$ with $\homoh(\varepsilon) = \domelemd$.
\end{lemma}
\begin{proof}
Let $\interU$ be an $n$-configuration unit with $\varepsilon \in \conceptL^{\interU}$ iff $\domelemd \in \conceptL^{\interI}$, and that interprets of all role and concept names outside $\rolesunit \cup \conceptsunit$ as empty sets. 
It is obvious that exactly one such unit exists.
In what follows, we are going to define a function $\homoh : \DeltaU \to \DeltaI$ inductively. 
Denoting the restriction of $\homoh$ to $\{ 0, 1 \}^{\leq k}$ by $\homoh_{\leq k}$, our inductive assumption states, for a given $k \leq n$, that $\homoh_{\leq i}$ is defined for all $i < k$ and $\homoh_{\leq i}$ is a homomorphism from $\restr{\interU}{\{ 0, 1 \}^{\leq k}}$ to $\interI$, \ie the restriction of $\interU$ to the set $\{ 0, 1 \}^{\leq k}$.

We first set $\homoh(\varepsilon) = \domelemd$. It is immediate to check that $\homoh_{\leq 0}$ is indeed a homomorphism (by $\interI \models$ \ref{axiom:allloopsbutnext} and our assumptions on $\varepsilon \in \conceptL^{\interU}$, and on empty interpretations of symbols outside $\rolesunit \cup \conceptsunit$).

For the inductive step, suppose that the assumption holds for some $1 \leq k \leq n$ and take a word~$\addresswordw \in \{ 0,1 \}^{k{-}1}$. 
We are going to define $\homoh(\addresswordw 0)$ as follows (the case of $\homoh(\addresswordw 1)$ is symmetric). 
Note that since $\homoh(\addresswordw) \in \conceptLvl_{k{-}1}^{\interI}$ (by the fact that $\homoh_{\leq k{-}1}$ is a homomorphism) and since $\interI \models$ \ref{axiom:LsuccLvli} (for $i$ equal to $k{-}1$) we conclude the existence of $\domelemd' \in \conceptLvl_{k}^{\interI}$ satisfying $(\homoh(\addresswordw),\domelemd') \in \roleell_k^{\interI}$. 
Note that also $\domelemd' \in \conceptL^{\interI} \cap (\conceptadr{k}^0)^{\interI}$ holds (by $\interI \models$ \ref{axiom:LsuccLvli} and $\interI \models$ \ref{axiom:LBitZeroi} with $i = k$).
Thus, we simply let $\homoh(\addresswordw 0) := \domelemd'$.

What remains to be shown is the fact that $\homoh_{\leq k}$ is a a homomorphism from $\restr{\interU}{\{ 0, 1 \}^{\leq k}}$ to $\interI$.
We already know that $\homoh_{\leq k{-}1}$ preserves concepts and roles, thus we can focus on concepts and roles involving words of length $k$.
Hence, take any $\addresswordw$ of length $k$ and proceed with concepts first.
Let $\conceptA$ be any concept name and assume that $\addresswordw \in \conceptA^{\interU}$. 
Our goal is to show that $\homoh(\addresswordw) \in \conceptA^{\interI}$.
The cases of $\conceptA \in \{ \conceptLvl_k, \conceptL, \conceptR, \conceptadr{k}^b \}$ follows immediately from the construction (see the discussion while defining them).
The cases of $\conceptA = \conceptLvl_j$ with $j \neq k$ and $\conceptA = \conceptadr{i}^b$ with $i \leq k$ cannot happen by the definition of $n$-configuration unit.
Thus the only cases left are these with $\conceptA = \conceptadr{i}^b$ with $i < k$.
But this is easy: let $\addresswordw = \addresswordu \addresswordv$ with $|\addresswordv| = 1$. 
By definition of $\interU$ we have that $\addresswordu \in (\conceptadr{i}^b)^{\interU}$. 
Since $\homoh_{\leq k{-}1}$ is a homomorphism, we infer $\homoh(\addresswordu) \in (\conceptadr{i}^b)^{\interI}$ and, by $\interI \models$ \ref{axiom:propbiti}, we conclude $\homoh(\addresswordw) \in (\conceptadr{i}^b)^{\interI}$.
Now we proceed with the case of role preservation by $\homoh$. 
Reasoning analogously, we may focus on roles $\roles$ from $\rolesunit$ and involving $\addresswordw$ only. 
Thus, by definition of $\interU$, the only cases are self-loops (that follows by $\interU \models$ \ref{axiom:allloopsbutnext}, \ref{axiom:leavesnextloop}) and the roles $\roleell_k^{\interU}$ and $\roler_k^{\interU}$ between the parent of $\addresswordw$ and $\addresswordw$, that follow from the construction.

This finishes the induction, implying that $\homoh$ is indeed a homomorphism from $\interU$ to $\interI$ satisfying~$\homoh(\varepsilon) = \domelemd$.
\end{proof}

At this point, we would like to give the reader some intuitions why units are decorated with different self-loops.
First, we show that their presence can be exploited to navigate top-down through a given unit.

 \begin{lemma}\label{lemma:from-root-to-leaf-path-via-loops}
 Let $\interU$ be an $n$-configuration unit. 
 Then for all $\addresswordw \in \DeltaU$ we have $(\varepsilon, \addresswordw) \in {\roleell_1}^{\interU} \circ {\roler_1}^{\interU} \circ \ldots \circ {\roleell_n}^{\interU} \circ {\roler_n}^{\interU}$ with ``$\circ$'' denoting the composition of relations, \ie $\roles^{\interU} \circ \rolet^{\interU} := \{ (\domelemc,\!\domeleme) \mid  (\domelemc,\!\domelemd) {\,\in\,} \roles^{\interU}  \text{ and } (\domelemd,\!\domeleme) {\,\in\,} \rolet^{\interU} \text{ for some } \domelemd\}$.
 \end{lemma}
 \begin{proof}
 For simplicity we use $\roles_i^{\interU}$ as an abbreviation of ${\roleell_1}^{\interU} \circ {\roler_1}^{\interU} \circ \ldots \circ {\roleell_i}^{\interU} \circ {\roler_i}^{\interU}$. 
 The proof is by induction, where the assumption is that for all $1 \leq i \leq n$ we have that all words $w$ of length at most $i$ satisfy $(\varepsilon, \addresswordw) \in \roles_i^{\interU}$.
 The base case (for $\addresswordw \in \{\varepsilon, 0, 1\}$) is immediate to verify. 
 Now take any word $\addresswordw$ of length at most $i{+}1$ and consider the following two cases:
 \begin{enumerate}\itemsep0em
   \item $|\addresswordw| \leq i$. Hence, by the inductive assumption we have $(\varepsilon, w) \in \roles_i^{\interU}$. 
   Since $(\addresswordw,\addresswordw) \in \roleell_{i{+}1}^{\interU}$ and $(\addresswordw,\addresswordw) \in \roler_{i{+}1}^{\interU}$, by the definition of composition we conclude $(\varepsilon, \addresswordw) \in \roles_{i{+}1}^{\interU}$.
   \item $|\addresswordw| = i{+}1$. Hence, $\addresswordw = \addresswordu0$ or $\addresswordu1$ for some $|\addresswordu| = i$. 
   We focus on the first case, the second one is symmetric.
   By inductive assumption we infer that $(\varepsilon, \addresswordu) \in \roles_i^{\interU}$. Since $(\addresswordu, \addresswordu0) \in \roleell_{i{+}1}^{\interU}$ and $(\addresswordw, \addresswordw) \in \roler_{i{+}1}^{\interU}$ we conclude $(\varepsilon, \addresswordw) \in \roles_{i{+}1}^{\interU}$, which finishes the proof. \qedhere
 \end{enumerate}
 \end{proof}

\noindent As a corollary of~\cref{lemma:from-root-to-leaf-path-via-loops}, we conclude that there is a \emph{single} CQ detecting root-leaf pairs in~units.

 \begin{corollary}\label{lemma:query-from-root-to-leaf}
     Let $\interU$ be an $n$-configuration unit. 
     There is a single conjunctive query $\qrl$ with $\varx_0, \varx_{2n} {\in} \queryVar{\qrl}$ such that the set $M = \{ (\matchpi(\varx_0), \matchpi(\varx_{2n})) \mid \interU \modelsmatch{\matchpi} \qrl  \}$ is equal to the set of root-leaf pairs from $\interU$, \ie $\conceptLvl_0^{\interU} \times \conceptLvl_n^{\interU}$.
 \end{corollary}
 \begin{proof}
   Take $\qrl := (\conceptLvl_0?;\roleell_1;\roler_1;\ldots;\roleell_n;\roler_n;\conceptLvl_n?)(\varx_0, \varx_{2n})$.
   The correctness follows from~\cref{lemma:from-root-to-leaf-path-via-loops}.
 \end{proof}

\section{From Units to Configuration Trees}\label{sec:conf-tress}
In the next step, we enrich $(\nofM{+}1)$-configuration units with additional concepts, allowing the units to represent a meaningful configuration of our ATM $\atmM = (\nofM, \statesQ, \statesQexists, \initialstate, \acceptingstate, \rejectingstate, \transreldelta )$. To this end, we employ a variety of new concept names from $\conceptsconf$ consisting of
\[
\conceptsconf := \big\{  \HeadHere, \NoHeadHere, \currstates, \HeadPos{i}^b, \HeadLetter{\lettera}, \conceptLetter{\lettera}, \conceptzero, \conceptone \mid \states \in \statesQ, b \in \set{0,1}, i \in \set{1,\ldots,\nofM}, \lettera \in \set{\letterzero,\letterone} \big\}.
\]

Before turning to a formal definition we first describe how configurations are structurally represented in models.
Recall that a configuration of $\atmM$ is a word $\tapewordw\states \tapewordw'$ with $|\tapewordw\tapewordw'| = 2^{\nofM}$ (called \emph{tape}) and $\states \in \statesQ$.
In our encoding, this configuration will be represented by an $(\nofM{+}1)$-configuration unit $\interC$ decorated by concepts from $\conceptsconf$. 
The interpretation $\interC$ stores the state $\states$, by associating the state concept $\currstates$ to its root.
The tape content $\tapewordw\tapewordw'$ is represented by the internal nodes of $\interC$: the $i$-th letter of $\tapewordw\tapewordw'$ (\ie the content of the ATM's $i$-th tape cell) is represented by the $i$-th node (according to their binary addresses) at the $\nofM$-th level. In case this letter is $\letterzero$, the corresponding node will be assigned the concept
$\conceptLetter{\letterzero}$, while $\letterone$ is represented by $\conceptLetter{\letterone}$. Yet, for reasons that will become clear only later, the tape cells' content is additionally represented in another way:
if it is $\letterzero$, then we label the $i$-th node's left child with $\conceptzero$ and its right child with $\conceptone$. The reverse situation is implemented when node represents the letter $\letterone$. 
Finally, there is a unique tape cell that is visited by the head of $\atmM$ and the node corresponding to this cell is explicitly marked by the concept $\HeadHere$ while all other ``tape cell nodes'' are marked by $\NoHeadHere$. In order to implement this marking correctly, the head's position's address needs to be explicitly recorded. Consequently, $\interC$'s root node stores this address (binarily encoded using the $\HeadPos{i}^b$ concepts) and from there, these concept assignments are broadcasted to and stored in all tape cell nodes on the $\nofM$-th level. 
Similarly, we decorate $\interC$'s root with the concept $\HeadLetter{\lettera}$ meaning that the current letter scanned by the head is~$\lettera$.
The structure just described can be visualised as follows:
\begin{figure}[H]
\centering
\includegraphics[scale=0.17]{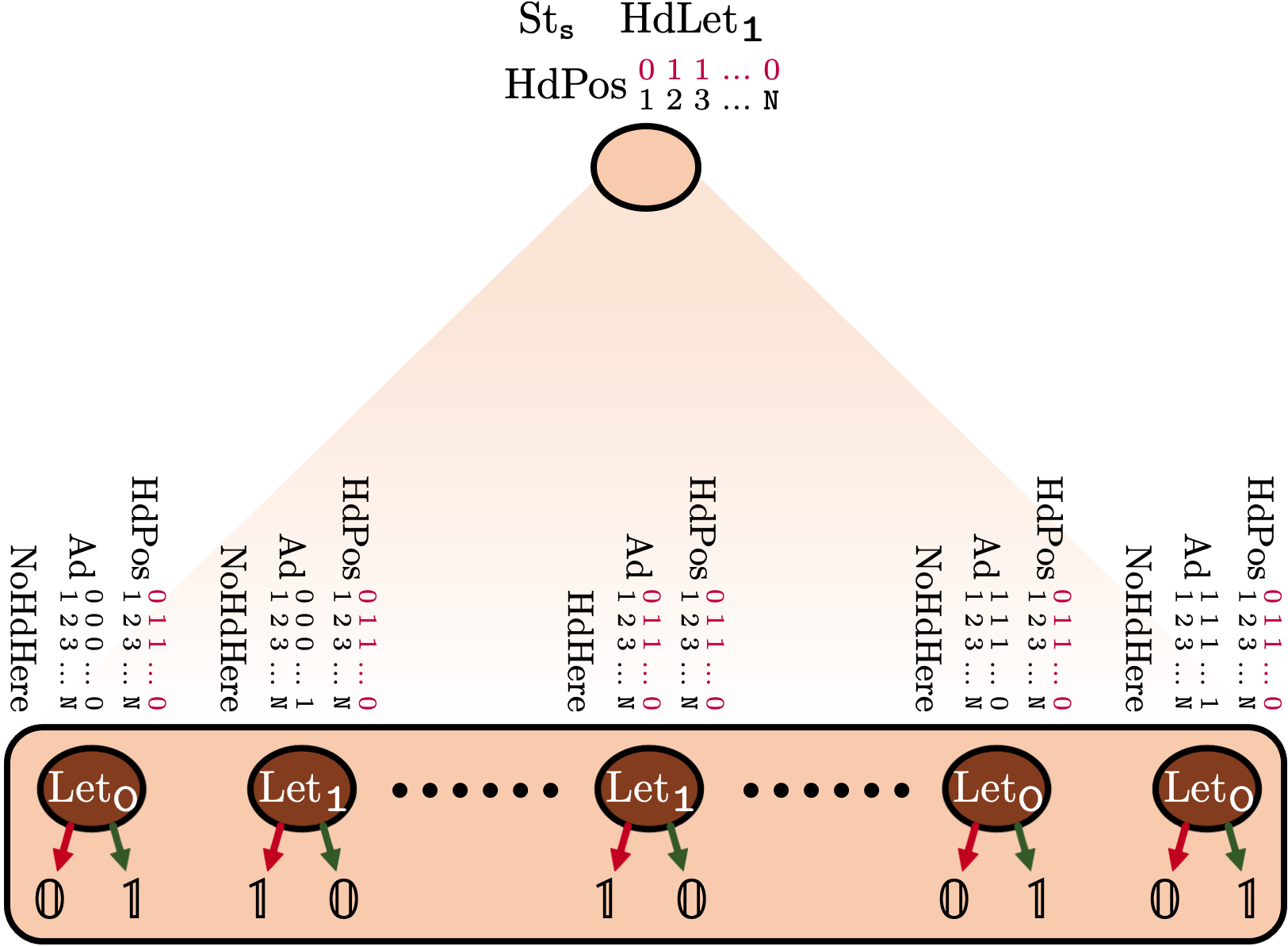}
\end{figure}

After this informal description and depiction, the formal definition of configuration trees should be plausible.

\begin{definition}[configuration tree]\label{def:conf-trees}
A \emph{configuration tree} $\interC$ of $\atmM$ is an interpretation $\interC = (\DeltaC,\cdot^\interC)$ such that $\interC$ is an $(\nofM{+}1)$-configuration unit additionally satisfying:
\begin{itemize}\itemsep0em
  \item There exists a unique state $\states \in \statesQ$ such that $(\currstates)^{\interC} = \{ \varepsilon \} $ and $(\currstate{\states'})^{\interC} = \emptyset$ for all $\states' \neq \states$.
  \item $(\conceptLvl_{\nofM{+}1})^{\interC} = \conceptzero^{\interC} \cup \conceptone^{\interC}$ and \  $\conceptzero^{\interC} \cap  \conceptone^{\interC} = \emptyset$.
  \item $(\conceptLetter{\letterzero})^{\interC} = \{ \addresswordw \mid \addresswordw0 \in \conceptzero^{\interC}, \addresswordw1 \in \conceptone^{\interC} \}$, $(\conceptLetter{\letterone})^{\interC} = \{ \addresswordw \mid \addresswordw0 \in \conceptone^{\interC}, \addresswordw1 \in \conceptzero^{\interC} \}$, and $(\conceptLetter{\letterzero})^{\interC} \cup (\conceptLetter{\letterone})^{\interC} = \conceptLvl_{\nofM}^{\interC}$. 
  \item There is a \emph{unique} word $\addresswordw_{\textit{head}}$ of length $\nofM$ witnessing $\HeadHere^{\interC} {=} \{ \addresswordw_{\textit{head}} \}$ and $\NoHeadHere^{\interC} {=} \conceptLvl_{\nofM}^{\interC} \setminus \{ \addresswordw_{\textit{head}} \}$.

  \item For $1 \leq i \leq \nofM$ and $b \in \{0,1\}$ s.t. $\addresswordw_{\textit{head}} \in (\conceptadr{i}^b)^{\interC}$ we~put\footnote{This is well-defined since for any $i$, we have that $\addresswordw_{\textit{head}}$ belongs to exactly one of $(\conceptadr{i}^0)^{\interC}, (\conceptadr{i}^1)^{\interC}$ by the definition of a unit.}
  $(\HeadPos{i}^b)^{\interC} = \conceptLvl_{0}^{\interC} {\cup} \conceptLvl_{\nofM}^{\interC} \; \text{and} \; (\HeadPos{i}^{1-b})^{\interC} = \emptyset$.
  \item $\HeadLetter{\lettera}^{\interC} = \{ \varepsilon \}$ and $\HeadLetter{\letterone-\lettera}^{\interC} = \emptyset$, where $\lettera$ is the unique letter from $\{ \letterzero, \letterone \}$ such that $\addresswordw_{\textit{head}} \in \conceptLetter{\lettera}^{\interC}$.\\
  \end{itemize}
\end{definition}

We next proceed with the corresponding axiomatisation.

\begin{enumerate}\itemsep0em
  \item To express that $\interC$ is an $(\nofM{+}1)$-configuration unit we integrate all the GCIs from $\kbKunit{\nofM{+}1}$.
  \item The root of $\interC$ is labelled with a unique state concept.
  \begin{multicols}{2}
  \begin{description}\itemsep0em
      \item[\desclabel{(StCov)}{axiom:StCov}] $\conceptLvl_{0} \equiv \bigdlor_{\states \in \statesQ} \currstates$
  \end{description}
  \columnbreak
  \begin{description}\itemsep0em
      \item[\desclabel{(StDisj[$\states$,$\states'$])}{axiom:StDisj}] $\currstates \dland \currstate{\states'} \dlsubseteq \botconcept \quad (\text{for all} \; \states \neq \states')$
  \end{description}
  \end{multicols}
  \item To axiomatise the coherent representation of the tape's content we employ:
  \begin{multicols}{2}
  \begin{description}\itemsep0em
      \item[\desclabel{(LetDisj)}{axiom:LetDisj}] \; $\conceptzero \dland \conceptone \dlsubseteq \botconcept$
      \item[\desclabel{(LetCov)}{axiom:LetCov}] $\conceptLvl_{\nofM{+}1} \equiv \conceptzero \dlor \conceptone$
  \end{description}
  \columnbreak
  \begin{description}\itemsep0em 
      \item[\desclabel{(LetConDisj)}{axiom:LetConDisj}] $\conceptLetter{\letterzero} \dland \conceptLetter{\letterone} \dlsubseteq \botconcept$
      \item[\desclabel{(LetConCov)}{axiom:LetConCov}]   $\conceptLetter{\letterzero} \dlor \conceptLetter{\letterone} \hspace{0.05em}\equiv \conceptLvl_{\nofM}$
  \end{description}
  \end{multicols}
  \begin{description}
      \centering
      \item[\desclabel{(EncLetZero)}{axiom:EncLetZero}] $\conceptLetter{\letterzero} \dlsubseteq \forall{\roleell_{\nofM{+}1}} ( \conceptLvl_{\nofM{+}1} \to \conceptzero) \dland \forall{\roler_{\nofM{+}1}} ( \conceptLvl_{\nofM{+}1} \to \conceptone)$
      \item[\desclabel{(EncLetOne)}{axiom:EncLetOne}] \;$\conceptLetter{\letterone} \dlsubseteq \forall{\roleell_{\nofM{+}1}} ( \conceptLvl_{\nofM{+}1} \to \conceptone) \dland \forall{\roler_{\nofM{+}1}} ( \conceptLvl_{\nofM{+}1} \to \conceptzero)$
  \end{description}
  \item Next, for the concepts $\HeadPos{1}^b, \ldots, \HeadPos{\nofM}^b$ we make sure they encode exactly one proper binary address (meant to encode the head's current position) in the root of $\interC$. Below we assume $1 \leq i \leq \nofM$.
  \begin{multicols}{2}
  \begin{description}\itemsep0em
      \item[\desclabel{(HdPosCov[i])}{axiom:HdPosCovi}] $\conceptLvl_{0} \dlor \conceptLvl_{\nofM} \equiv \HeadPos{i}^0 \dlor \HeadPos{i}^1$
  \end{description}
  \columnbreak
  \begin{description}\itemsep0em
      \item[\desclabel{(HdPosDisj[i])}{axiom:HdPosDisji}] $\HeadPos{i}^0 \dland \HeadPos{i}^1 \dlsubseteq \botconcept $
  \end{description}
  \end{multicols}
  \item Another step is to propagate the head address stored in the root to all nodes on the $\nofM$-th level of $\interC$.
    Here we exploit the presence of self-loops and~\cref{lemma:from-root-to-leaf-path-via-loops}, and use the following GCIs (for $1 \leq i \leq \nofM$ and~$b \in \{ 0,1 \}$):\footnote{We note that the same can be achieved without exploitation of self-loops by iteratively propagating the $\HeadPos{i}^b$ through the tree, but the first author believes that the presented formulation is elegant and makes the reader get used to the presence self-loops.}
  \begin{description}\itemsep0em
      \item[\desclabel{(PropHdPos[i,b])}{axiom:PropHdPosib}] $\conceptLvl_0 \dland \HeadPos{i}^b \dlsubseteq \forall{\roleell_1}\forall{\roler_1}\ldots\forall{\roleell_{\nofM}}\forall{\roler_{\nofM}} \; (\conceptLvl_\nofM \to \HeadPos{i}^b)$
  \end{description}
  \item We distinguish between the node representing the cell visited by the head (assigning $\HeadHere$) and the other cell nodes (assigning $\NoHeadHere$) by having every cell node compare their address (stored in the $\conceptadr{i}^b$ concepts) with the head address received through the broadcast from the root.
  \begin{description}\itemsep0em
      \item[\desclabel{(HdHereCov)}{axiom:HdHereCov}] \qquad \qquad \qquad \qquad \qquad \; \; \; $\HeadHere \dlor \NoHeadHere \equiv \conceptLvl_{\nofM}$
      \item[\desclabel{(HdHereEqualAdr)}{axiom:HdHereEqualAdr}] \;\,\, $\conceptLvl_{\nofM} \dland \bigdland_{i=1}^{\nofM} \bigdlor_{b\in\{0,1\}} \big( \conceptadr{i}^b \dland \HeadPos{i}^b \big) \; \; \; \, \dlsubseteq \HeadHere$
      \item[\desclabel{(NoHdHereDiffrAdr)}{axiom:NoHdHereDiffrAdr}] $\conceptLvl_{\nofM} \dland \bigdlor_{i=1}^{\nofM} \bigdlor_{b\in\{0,1\}}  \big( \conceptadr{i}^b \dland \HeadPos{i}^{1-b} \big) \dlsubseteq \NoHeadHere$      
  \end{description}
  \item We synchronise the letter scanned by the head of $\atmM$ with its ``recording'' in the root (below $\lettera \in \{ \letterzero, \letterone \}$).
  \begin{description}\itemsep0em
      \item[\desclabel{(HdLetCov)}{axiom:HdLetCov}] \qquad \qquad \qquad \qquad \qquad \, \, $\HeadLetter{\letterzero} \dlor \HeadLetter{\letterone} \equiv \conceptLvl_0$
      \item[\desclabel{(RetrHdLet[$\lettera$])}{axiom:RetrHdLet}] $\conceptLvl_0 \dland \exists{\roleell_1}\exists{\roler_1} \ldots \exists{\roleell_{\nofM}}\exists{\roler_{\nofM}} (\HeadHere \dland \conceptLetter{\lettera}) \: \dlsubseteq \HeadLetter{\lettera}$
      \item[\desclabel{(HdLetUnique[$\lettera$])}{axiom:HdLetUnique}] \quad \quad \quad \quad \quad \quad \quad \quad \quad \, \,  $\conceptLvl_0 \dland \HeadLetter{\lettera} \dlsubseteq \forall{\roleell_1}\forall{\roler_1} \ldots \forall{\roleell_{\nofM}}\forall{\roler_{\nofM}} (\HeadHere \to \conceptLetter{\lettera})$
  \end{description}
\end{enumerate}

This finishes the axiomatisation of configuration trees. 

\noindent For the knowledge base $\kbKconf$, composed of all presented GCIs, we present its correctness in the following lemmas.
Similarly to the previous section, both of them are routine and the reader may omit them at first reading.

\begin{lemma}\label{prop:conf-trees-satisfies-kbKconf}
Any configuration tree $\interC$ is a model of $\kbKconf$.
\end{lemma}
\begin{proof} 
    Let $\interC$ be a configuration tree. We proceed with all axioms $\alpha$ of $\kbKconf$ and show that $\interC \models \alpha$.
    \begin{enumerate}\itemsep0em
    \item Since $\interC$ is an $(\nofM{+}1)$-configuration unit by definition, by~\cref{prop:n-conf-units-satisfies-KBKunit} we infer $\interC \models \kbKunit{\nofM{+}1}$.
    \item The satisfactions $\interC \models$ \ref{axiom:StCov} and $\interC \models$ \ref{axiom:StDisj} follow immediately from the first item of~\cref{def:conf-trees}.
    \item We have $\interC \models$ \ref{axiom:LetDisj} and $\interC \models$ \ref{axiom:LetCov} by the second item of~\cref{def:conf-trees}.
    Next, we have $\interC \models$ \ref{axiom:LetConDisj} by the fact that words in $\conceptLetter{\letterzero}^{\interC}$ and $\conceptLetter{\letterone}^{\interC}$ have different last letters.
    The satisfaction of \ref{axiom:LetConCov} by $\interC$ is due to the 3rd property in the 3rd item of~\cref{def:conf-trees}.
    What remains to show is the satisfaction of \ref{axiom:EncLetZero} (the proof of \ref{axiom:EncLetOne} is symmetric), but this is due to the fact that the only $\roleell_{\nofM{+}1}$- (resp. $\roler_{\nofM{+}1}$-) successor of any $\addresswordw \in \conceptLetter{\letterzero}^{\interC}$ (thus also from $\conceptLvl_{\nofM}^{\interC}$ by the previous statement) is $\addresswordw 0$ (resp. $\addresswordw1$) that belong to $\conceptzero^{\interC}$ (resp. $\conceptone^{\interC}$) by definition.
    \item Before we proceed with the next part of the axiomatisation, we establish a few properties about concept memberships of $\addresswordw_{\textit{head}}$.
    First, we have $\addresswordw_{\textit{head}} \in \conceptLvl_\nofM^{\interC}$ by definition (4th item).
    Moreover, for every $i$ we have that $\addresswordw_{\textit{head}}$ belongs to exactly one of $(\conceptadr{i}^0)^{\interC}, (\conceptadr{i}^1)^{\interC}$ by the definition of a unit. 
    Hence, the choice of $\addresswordw_{\textit{head}}$ fixes interpretations of $\conceptadr{i}^0, \conceptadr{i}^1$ and, as we will see, also $\HeadPos{i}^0$ and $\HeadPos{i}^1$.
    Indeed, by the 5th item of~\cref{def:conf-trees}, whenever $\addresswordw_{\textit{head}} \in (\conceptadr{i}^{b})^{\interC}$ then $(\HeadPos{i}^{b})^{\interC} = \conceptLvl_0^{\interC} \cup \conceptLvl_\nofM^{\interC}$ and $(\HeadPos{i}^{1-b})^{\interC} = \emptyset$, thus also $(\HeadPos{i}^{b})^{\interC} \cap (\HeadPos{i}^{1-b})^{\interC} = \emptyset$ and $(\HeadPos{i}^{b})^{\interC} \cup (\HeadPos{i}^{1-b})^{\interC} = \conceptLvl_0^{\interC} \cup \conceptLvl_\nofM^{\interC}$ hold.
    This establishes $\interC \models$ \ref{axiom:HdPosDisji} and $\interC \models$ \ref{axiom:HdPosCovi} for all $1 \leq i \leq \nofM$.
    \item If $\varepsilon \in (\conceptLvl_0 \dland \HeadPos{i}^{b})^{\interC}$, then by definition $\addresswordw_{\textit{head}} \in (\conceptadr{i}^{b})^{\interC}$. This implies $(\HeadPos{i}^{b})^{\interC} = \conceptLvl_0^{\interC} \cup \conceptLvl_\nofM^{\interC}$ and thus $\conceptLvl_\nofM^{\interC} \subseteq (\HeadPos{i}^{b})^{\interC}$, which is even stronger than the meaning of the GCI \ref{axiom:PropHdPosib}.
    Hence, we have $\interC \models$ \ref{axiom:PropHdPosib} for all $1 \leq i \leq \nofM$ and $0 \leq b \leq 1$.
    \item We next focus on proving $\interC \models$ \ref{axiom:HdHereEqualAdr} (the proof of \ref{axiom:NoHdHereDiffrAdr} is symmetric) and $\interC \models$ \ref{axiom:HdHereCov}. The second GCI follows by definition, hence we focus on the first one.
    Ad absurdum, assume that there is $\addresswordw \in \textstyle (\conceptLvl_{\nofM} \dland \bigdland_{i=1}^{\nofM} \bigdlor_{b\in\{0,1\}} \big( \conceptadr{i}^b \dland \HeadPos{i}^b \big))^{\interC}$ but $\addresswordw \not\in (\HeadHere)^{\interC}$.
    Since $\addresswordw \in (\conceptLvl_{\nofM})^{\interC}$ we infer $|\addresswordw| = \nofM$ and, by $\addresswordw \not\in (\HeadHere)^{\interC}$ and the 4th item of~\cref{def:conf-trees}, we infer $\addresswordw \neq \addresswordw_{\textit{head}}$.
    Thus there is a position $1 \leq k \leq \nofM$ such the $k$-th letter of $\addresswordw$ differs from the $k$-th letter of $\addresswordw_{\textit{head}}$ (called it $b$).
    So we have $\addresswordw \not\in (\conceptadr{k}^{b})^{\interC}$ and $\addresswordw \in (\conceptadr{k}^{1-b})^{\interC}$ (by definition of $\conceptadr{i}^{b}$ in units).
    At the same time the 5th item of~\cref{def:conf-trees} informs us that $(\HeadPos{k}^{b})^{\interC} = \conceptLvl_0^{\interC} \cup \conceptLvl_\nofM^{\interC}$ and $(\HeadPos{k}^{1-b})^{\interC} = \emptyset$, which implies $\addresswordw \in (\HeadPos{k}^{b})^{\interC}$ and $\addresswordw \not\in (\HeadPos{k}^{1-b})^{\interC}$
    This contradicts the fact that $\addresswordw \in (\bigdlor_{b\in\{0,1\}} \big( \conceptadr{i}^b \dland \HeadPos{i}^b \big))^{\interC}$.
    \item We proceed with the last three GCIs. 
    Satisfaction of \ref{axiom:HdLetCov} by $\interC$ follows by definition.
    For~\ref{axiom:RetrHdLet}, assume that its antecedent is non-empty (which means that it is equal to $\{ \varepsilon \}$). 
    This implies, by the 4th and the last item of~\cref{def:conf-trees}, that $\addresswordw_{\textit{head}} \in \conceptLetter{\lettera}^{\interC}$. 
    Thus, by definition, $\HeadLetter{\lettera}^{\interC}$ is equal to $\{ \varepsilon \}$, which obviously contains $\varepsilon$. 
    Hence, $\interC\models$  \ref{axiom:RetrHdLet}.
    Finally, $\interC \models $\ref{axiom:HdLetUnique} is shown as follows.
    If the antecedent of \ref{axiom:HdLetUnique} is non-empty then it is equal to $\{ \varepsilon \}$. 
    Thus by the list item of~\cref{def:conf-trees}, we have $\conceptLetter{\lettera}^{\interC} = \{ \addresswordw_{\textit{head}} \}$. 
    Since $\HeadHere^{\interC} = \{ \addresswordw_{\textit{head}} \}$ (by the 4th item of~\cref{def:conf-trees}), we conclude $\HeadHere^{\interC} \subseteq \conceptLetter{\lettera}^{\interC}$. Therefore $\interC \models$ \ref{axiom:HdLetUnique} holds, finishing the proof. \qedhere
  \end{enumerate}
\end{proof}

\begin{lemma}\label{lemma:conf-trees-homomorphisms}
For any model $\interI$ of $\kbKconf$ and any $\domelemd \in \conceptLvl_0^{\interI}$ there is a configuration tree $\interC$ and a homomorphism~$\homoh$ from $\interC$ into $\interI$ with $\homoh(\varepsilon) = \domelemd$.
\end{lemma}
\begin{proof}
    By~\cref{lemma:n-conf-units-homomorphisms} there is an $(\nofM{+}1)$-configuration unit $\interU$ and a homomorphism $\homoh : \interU \to \interI$ with $\homoh(\varepsilon) = \domelemd$. 
    Moreover, as the symbols outside $\rolesunit \cup \conceptsunit$ do not appear in~\cref{def:conf-units} we can assume that $\interU$ interprets them as empty sets.
    Let $\interC = (\DeltaU, \cdot^{\interC})$ be an interpretation that is obtained from changing the meaning of concepts from $\conceptsconf$ as follows: for any $\conceptC \in \conceptsconf$ we let $\conceptC^{\interC} := \{ \addresswordw \mid \homoh(\addresswordw) \in \conceptC^{\interU} \}$.
    All other symbols are interpreted as in $\interU$.
    Clearly $\homoh$ is a homomorphism from $\interC$ into $\interI$ with $\homoh(\varepsilon) = \domelemd$.
    It suffices to prove that $\interC$ is a configuration tree. 
    This is done by routine investigation of items from~\cref{def:conf-trees} and the presented GCIs.
    \begin{itemize}\itemsep0em
    \item Let $\states$ be the unique state satisfying $\homoh(\varepsilon) \in \currstates^{\interI}$: it exists by $\interI \models$ \ref{axiom:StCov} and is unique by~$\interI \models$ \ref{axiom:StDisj}. 
    Hence, $(\currstates)^{\interC} = \{ \varepsilon \}$ holds.
    Moreover, $(\currstate{\states'})^{\interC} = \emptyset$ for all $\states' \neq \statesQ$.
    Note that $\varepsilon \not\in (\currstate{\states'})^{\interC}$ by \ref{axiom:StDisj} and for other elements (so from $\conceptLvl_i^{\interC}$ for some $i > 0$) their membership in $(\currstate{\states'})^{\interC}$ would violate $\interI \models \textstyle \bigdlor_{\states \in \statesQ} \currstates$.

    \item Similarly, the equalities $(\conceptLvl_{\nofM{+}1})^{\interC} = \conceptzero^{\interC} \cup \conceptone^{\interC}$ and \  $\conceptzero^{\interC} \cap  \conceptone^{\interC} = \emptyset$ follow by $\interI \models$ \ref{axiom:LetDisj} and  $\interI \models$ \ref{axiom:LetCov}.

    \item Analogously, we have $(\conceptLetter{\letterzero})^{\interC} \cup (\conceptLetter{\letterone})^{\interC} = \conceptLvl_{\nofM}^{\interC}$ by $\interI \models$ \ref{axiom:LetConDisj} and $\interI \models$ \ref{axiom:LetConCov}. 
    Note that this implies $(\conceptLetter{\letterzero})^{\interC} \subseteq \conceptLvl_{\nofM}^{\interC}$.
    We next show the equality $(\heartsuit): (\conceptLetter{\letterzero})^{\interC} = \{ \addresswordw  \in \DeltaU \mid \addresswordw0 \in \conceptzero^{\interC}, \addresswordw1 \in \conceptone^{\interC} \}$ (the related equality for $\conceptLetter{\letterone}$ is symmetric).
    Take any $\addresswordw \in (\conceptLetter{\letterzero})^{\interC}$ (thus also $\in \conceptLvl_{\nofM}^{\interC}$).
    By the fact that $\interC$ is a unit, we infer that $\addresswordw0 \in \conceptLvl_{\nofM{+}1}^{\interC}$ and $\addresswordw1 \in \conceptLvl_{\nofM{+}1}^{\interC}$ exist, and moreover $\addresswordw$ is linked to them, respectively, by the roles $\roleell_{\nofM{+}1}^{\interC}$ and $\roler_{\nofM{+}1}^{\interC}$.
    Hence, by the homomorphic assignment of concepts from $\conceptsconf$ and the satisfaction $\interI \models$ \ref{axiom:EncLetZero} we have that $\addresswordw0 \in \conceptzero^{\interC}$ and $\addresswordw1 \in \conceptone^{\interC}$. Hence, the $\subseteq$-relationship of $(\heartsuit)$ follows.
    For the $\supseteq$-relationship take any $\addresswordw \in \DeltaU$ s.t. $\addresswordw0 \in \conceptzero^{\interC}, \addresswordw1 \in \conceptone^{\interC}$ and note that $\addresswordw \in \conceptLvl_{\nofM}^{\interC}$ and $\addresswordw b \in \conceptLvl_{\nofM{+}1}^{\interC}$ hold.
    Otherwise, by the fact that $\interC$ is an $(\nofM{+}1)$-configuration unit, the element $\addresswordw0$ does not exists or it belongs to $\conceptLvl_{i}^{\interC}$ for $i \neq \nofM{+}1$, which violates $\interI \models$ \ref{axiom:LetCov}.
    By $\interI \models$ \ref{axiom:LetConCov} we know that $\addresswordw \in (\conceptLetter{\letterzero})^{\interC} \cup (\conceptLetter{\letterone})^{\interC}$.
    If $\addresswordw \in (\conceptLetter{\letterzero})^{\interC}$ holds then we are done.
    Thus, assume towards a contradiction that $\addresswordw \in (\conceptLetter{\letterone})^{\interC}$.
    But then the first conjunct of the consequent of \ref{axiom:EncLetOne} is violated, contradicting its satisfaction by $\interI$.
    Hence~$(\heartsuit)$ holds.

    \item Let $\addresswordw_{\textit{head}}$ be the unique $\nofM$-digit binary word whose $i$-th letter is equal to $b$ iff $\varepsilon \in (\HeadPos{i}^b)^{\interC}$ holds. This is well defined due to $\interI \models$ \ref{axiom:HdPosDisji} and $\interI \models$ \ref{axiom:HdPosCovi} for all $1 \leq i \leq \nofM$.
    It remains to show that this $\addresswordw_{\textit{head}}$ indeed plays the role of $\addresswordw_{\textit{head}}$ in the sense of~\cref{def:conf-trees}.
    Take any $1 \leq i \leq \nofM$ and let $b$ be the $i$-th letter of $\addresswordw_{\textit{head}}$. 
    We will show that $(\HeadPos{i}^{1-b})^{\interC} = \emptyset$ holds. 
    Ad absurdum, assume that it is non-empty and contains $\addresswordw$.
    By $\interI \models$ \ref{axiom:HdPosCovi} we have that either $\addresswordw \in \conceptLvl_0^{\interC}$ or $\addresswordw \in \conceptLvl_\nofM^{\interC}$. 
    The first case is not possible due to $\interI \models$ \ref{axiom:HdPosDisji} (we already have that $\varepsilon \in (\HeadPos{i}^b)^{\interC}$).
    Thus $\addresswordw \in \conceptLvl_\nofM^{\interC}$. 
    But then it violates $\interI \models$ \ref{axiom:PropHdPosib}, which by~\cref{lemma:from-root-to-leaf-path-via-loops} enforces that $\addresswordw \in (\HeadPos{i}^{b})^{\interC}$.
    Hence, $(\clubsuit{:})\; (\HeadPos{i}^{1-b})^{\interC} = \emptyset$, which by \ref{axiom:HdPosCovi} implies $(\clubsuit'{:})\; \conceptLvl_0^{\interC} \cup \conceptLvl_{\nofM}^{\interC} = (\HeadPos{i}^{b})^{\interC}$.
    Thus it follows, by definition of $\addresswordw_{\textit{head}}$, that $\addresswordw_{\textit{head}} \in (\conceptadr{i}^b)^{\interC}$ iff $\addresswordw_{\textit{head}} \in  (\HeadPos{i}^{b})^{\interC}$, concluding that $\interC$ satisfies the 5th item of~\cref{def:conf-trees}.
    Next, by exploiting $(\clubsuit)$ and $(\clubsuit')$ and applying the satisfaction of \ref{axiom:HdHereEqualAdr}, \ref{axiom:NoHdHereDiffrAdr} and \ref{axiom:HdHereCov} by $\interC$, we conclude the satisfaction of the 4th item of~\cref{def:conf-trees}.
    This, among the other properties, results in $\HeadHere^{\interC} = \{ \addresswordw_{\textit{head}} \}$.
    Finally, let $\lettera$ be the unique letter satisfying $\addresswordw_{\textit{head}} \in \conceptLetter{\lettera}^{\interC}$ (it exists and is unique by $\interI \models$ \ref{axiom:LetConDisj}, \ref{axiom:LetConCov}). 
    We claim that $\HeadLetter{\lettera}^{\interC} = \{ \varepsilon \}$ and $\HeadLetter{\letterone-\lettera}^{\interC} = \emptyset$.
    Note that both $\HeadLetter{\lettera}^{\interC}$ and $\HeadLetter{\letterone-\lettera}^{\interC}$ are subsets of $\{ \varepsilon \}$ by $\interI \models $ \ref{axiom:HdLetCov}, thus it suffices to show that $\varepsilon \in \HeadLetter{\lettera}^{\interC}$  and $\varepsilon \not\in \HeadLetter{\letterone-\lettera}^{\interC}$.
    The first property holds due to $\interI \models$ \ref{axiom:RetrHdLet}. 
    For the second property, towards a contradiction assume that $\varepsilon \in \HeadLetter{\letterone-\lettera}^{\interC}$.
    Hence, by $\interI \models$ \ref{axiom:HdLetUnique} we conclude that $\HeadHere^{\interC} = \{ \addresswordw_{\textit{head}} \} \subseteq \conceptLetter{\letterone-\lettera}^{\interC}$.
    But $\conceptLetter{\letterone-\lettera}^{\interC}$ is empty by $\interI \models$ \ref{axiom:LetConDisj} and the definition of $\lettera$. A~contradiction.
  \end{itemize}

  Hence the interpretation $\interC$ is indeed a configuration tree, concluding the proof.
\end{proof}

\section{Enriching Configuration Trees}\label{sec:enriched-conf-tress}

Recall that the purpose of configuration trees is to place them inside a model that describes the run of the Turing machine $\atmM$. 
In particular, this will require to describe the progression of one configuration to another. 
In order to prepare for that, we next introduce an extended version of configuration trees that are enriched by additional information pertaining to their predecessor configuration in a run. 
To this end, we use new concept~names~from 
\[
\conceptsenrct := \big\{\prevtrans{\transitiont}, \InitialConf,  \PrevHeadHere, \NoPrevHeadHere, \PrevHeadAbove, \NoPrevHeadAbove, \PHeadPos{i}^b, \PHeadLetter{\lettera} \big\},
\]
with $\transitiont \in \transreldelta$, $1 \leq i \leq \nofM$, $b \in \{ 0,1 \}$, and $\lettera \in \{ \letterzero, \letterone \}$.
We use $\conceptsptrans$ to denote the set $\{ \InitialConf, \prevtrans{\transitiont} \mid \transitiont \in \transreldelta\}$.

The concept $\prevtrans{\transitiont}$, assigned to the root, indicates the transition, through which the configuration has been reached from the previous configuration, while $\InitialConf$ is used as its replacement for the initial configuration.
In addition, concepts $\PHeadPos{i}^b$ and $\PHeadLetter{\lettera}$ are attached to the root in order to --- in a way very similar to $\HeadPos{i}^b$ and $\HeadLetter{\lettera}$ --- indicate the previous configuration's head position as well as the letter stored in that position \emph{on the current configuration's tape}.
For the sake of our encoding we also employ the concepts $\PrevHeadHere, \NoPrevHeadHere$ that play the role analogous to the $\HeadHere$ and $\NoHeadHere$ concept from configuration-trees.\footnote{For simplicity of axiomatisation, the initial configuration will also carry previous head information, but it will be irrelevant.} 
For technical reasons, we also introduce the concepts $\PrevHeadAbove$ and $\NoPrevHeadAbove$ that will label nodes on the $(\nofM{+}1)$-th level iff their parent is labelled with the corresponding concept from $\{\PrevHeadHere, \NoPrevHeadHere\}$. 

We proceed with the formal definition of the structures just described.

\begin{definition}[enriched configuration tree]\label{def:enr-conf-tree}
An \emph{enriched configuration tree} $\interE$ of $\atmM$ is an interpretation~$\interE = (\DeltaE,\cdot^\interE)$ such that $\interE$ is a configuration tree additionally satisfying the following conditions on concepts from~$\conceptsenrct$:
\begin{itemize}\itemsep0em
\item There is exactly one concept $\conceptC \in \conceptsptrans$ for which $\conceptC^{\interE} = \{ \varepsilon \}$ and for all $\conceptC' \in \conceptsptrans \setminus \{ \conceptC \}$ we have $(\conceptC')^{\interE} = \emptyset$.
\item $\prevtrans{(\states,\lettera,\letterb,\states',d)}^{\interE} = \{ \varepsilon \}$ implies $(\currstate{\states'})^{\interE} =\{ \varepsilon \}$ for all transitions $(\states,\lettera,\letterb,\states',d) \in \transreldelta$.
\item $\PrevHeadHere^{\interE} {=} \{ \addresswordw_{\textit{phd}} \}$ and $\NoPrevHeadHere^{\interE} {=} \conceptLvl_{\nofM}^{\interE} \setminus \{ \addresswordw_{\textit{phd}} \}$ for the $\nofM$-digit binary word $\addresswordw_{\textit{phd}}$ encoding\footnote{Here we use the fact that $\atmM$ never attempts to move left (resp. right) on the left-most (resp. right-most) tape cell.}
    \begin{itemize}\itemsep0em
    \item the number obtained as $\addresswordw_{\textit{head}}-d$ (see: \cref{def:conf-trees}) whenever $\prevtrans{(\states,\lettera,\letterb,\states',d)}^{\interE} = \{ \varepsilon \}$, or 
    \item the number $0$ in case $\InitialConf^{\interE} = \{ \varepsilon \}$.
    \end{itemize}
\item $\PrevHeadAbove^{\interE} = \{ \addresswordw0, \addresswordw1 \mid \addresswordw \in \PrevHeadHere^{\interE} \}$ and $\NoPrevHeadAbove^{\interE} = \conceptLvl_{\nofM{+}1}^{\interE} \setminus \PrevHeadAbove^{\interE}$.
\item $(\PHeadPos{i}^b)^{\interE} = \conceptLvl_{0}^{\interE} {\cup} \conceptLvl_{\nofM}^{\interE} \; \text{and} \; (\PHeadPos{i}^{1-b})^{\interE} = \emptyset$ for all $1 \leq i \leq \nofM$ and $0 \leq b \leq 1$ with $\addresswordw_{\textit{phd}} \in (\conceptadr{i}^b)^{\interE}$.
\item $\PHeadLetter{\lettera}^{\interE} = \{ \varepsilon \}$ and $\PHeadLetter{\letterone-\lettera}^{\interE} = \emptyset$, where $\lettera$ is the unique letter from $\{ \letterzero, \letterone \}$ such that $\addresswordw_{\textit{phd}} \in \conceptLetter{\lettera}^{\interE}$.
\item $\InitialConf^{\interE} {=} \{ \varepsilon \}$ implies $\varepsilon \in \conceptL^{\interE}$, $\currstate{\initialstate}^{\interE} {=} \{ \varepsilon \}$, $\conceptLetter{\letterzero}^{\interE} = \conceptLvl_{\nofM}^{\interE}$, and $\HeadPos{i}^0 = \PHeadPos{i}^0 = \conceptLvl_{0}^{\interE} \cup \conceptLvl_{\nofM}^{\interE}$ for all~$1 \leq i \leq \nofM$.
\end{itemize}
\end{definition}

As usual, we supplement the above definition with the corresponding axiomatisation.
\begin{enumerate}\itemsep0em 
    \item We ensure, that the root unambiguously indicates the previous transition (or initiality). 
    Below $\transitiont \neq \transitiont' \in \transreldelta$.
    \begin{description}\itemsep0em
        \centering
        \item[\desclabel{(TrCov)}{axiom:TrCov}] $\conceptLvl_0 \; \equiv \; \InitialConf \dlor \textstyle\bigdlor_{\transitiont \in \transreldelta} \prevtrans{\transitiont}$,
    \end{description}
    \begin{multicols}{2}
    \begin{description}\itemsep0em
        \item[\desclabel{(TrInitDisj[$\transitiont$])}{axiom:TrInitDisjt}] $\InitialConf \dland \prevtrans{\transitiont} \dlsubseteq \botconcept$,
    \end{description}
        \columnbreak
    \begin{description}\itemsep0em
        \item[\desclabel{(TrDisj[$\transitiont, \transitiont'$])}{axiom:TrInitDisjttprim}] $\prevtrans{\transitiont} \dland \prevtrans{\transitiont'} \dlsubseteq \botconcept$.
    \end{description}
    \end{multicols}
    \item We provide the encoding of the previous head position and the previous letter scanned by the head. This is done by means of the $\PHeadPos{i}^b$, $\PHeadLetter{\lettera}, \PrevHeadHere$, and $\NoPrevHeadHere$ concepts in analogy to how it was done for the current head position (see the last four points of the axiomatisation from the previous section).
    Below we assume $1 \leq i \leq \nofM$, $b \in \{ 0,1 \}$, and $\lettera \in \{ \letterzero, \letterone \}$.
    \begin{multicols}{2}
    \begin{description}\itemsep0em
        \item[\desclabel{(PHdPosCov[i])}{axiom:PHdPosCovi}] $\conceptLvl_{0} {\dlor} \conceptLvl_{\nofM} \equiv \PHeadPos{i}^0 {\dlor} \PHeadPos{i}^1$,
    \end{description}
    \columnbreak
    \begin{description}\itemsep0em
        \item[\desclabel{(PHdPosDisj[i])}{axiom:PHdPosDisji}] $\PHeadPos{i}^0 \dland \PHeadPos{i}^1 \dlsubseteq \botconcept$,
    \end{description}
    \end{multicols}
    \begin{description}\itemsep0em
        \item[\desclabel{(PropPHdPos[i,b])}{axiom:PropPHdPosib}] $\conceptLvl_0 \dland \PHeadPos{i}^b \dlsubseteq \forall{\roleell_1}\forall{\roler_1}\ldots\forall{\roleell_{\nofM}}\forall{\roler_{\nofM}} \; (\conceptLvl_\nofM \to \PHeadPos{i}^b)$,
    \end{description}
  \begin{description}\itemsep0em
      \item[\desclabel{(PHdHereCov)}{axiom:PHdHereCov}] \hspace{10.4em} $\PrevHeadHere \dlor \NoPrevHeadHere \equiv \conceptLvl_{\nofM}$
      \item[\desclabel{(PHdHereEqualAdr)}{axiom:PHdHereEqualAdr}] \; $\conceptLvl_{\nofM} \dland \bigdland_{i=1}^{\nofM} \bigdlor_{b\in\{0,1\}} \big( \conceptadr{i}^b \dland \PHeadPos{i}^b \big) \; \; \, \, \dlsubseteq \PrevHeadHere$,
      \item[\desclabel{(NoPHdHereDiffAdr)}{axiom:NoPHdHereDiffrAdr}] $\conceptLvl_{\nofM} \dland \bigdlor_{i=1}^{\nofM} \bigdlor_{b\in\{0,1\}}  \big( \conceptadr{i}^b \dland \PHeadPos{i}^{1-b} \big) \dlsubseteq \NoPrevHeadHere$,
  \end{description}
  \begin{description}\itemsep0em
      \item[\desclabel{(PHdLetCov)}{axiom:PHdLetCov}] \hspace{12.2em} $\PHeadLetter{\letterzero} \dlor \PHeadLetter{\letterone} \equiv \conceptLvl_0$,
      \item[\desclabel{(RetrPHdLet[$\lettera$])}{axiom:RetrPHdLet}] \hspace{1.85em}  $\conceptLvl_0 \dland \exists{\roleell_1}\exists{\roler_1} \ldots \exists{\roleell_{\nofM}}\exists{\roler_{\nofM}} (\PrevHeadHere \dland \conceptLetter{\lettera}) \: \dlsubseteq \PHeadLetter{\lettera}$,
      \item[\desclabel{(PHdLetUnique[$\lettera$])}{axiom:PHdLetUnique}] \hspace{11.83em}  $\conceptLvl_0 \dland \PHeadLetter{\lettera} \dlsubseteq \forall{\roleell_1}\forall{\roler_1} \ldots \forall{\roleell_{\nofM}}\forall{\roler_{\nofM}} (\PrevHeadHere \to \conceptLetter{\lettera})$.
  \end{description}
  \item Next, the concepts $\PrevHeadAbove$ and $\NoPrevHeadAbove$ are assigned via
    \begin{multicols}{2}
    \begin{description}\itemsep0em
        \item[\desclabel{(PHdAbvCov)}{axiom:PHdAbvCov}] $\PrevHeadAbove \dlor \NoPrevHeadAbove \equiv \conceptLvl_{\nofM{+}1}$,
    \end{description}
    \columnbreak
    \begin{description}\itemsep0em
        \item[\desclabel{(PHdAbvDisj)}{axiom:PHdAbvDisj}] $\PrevHeadAbove \dland \NoPrevHeadAbove \dlsubseteq \botconcept$,
    \end{description}
    \end{multicols} 

    \begin{description}\itemsep0em
        \item[\desclabel{(PropPHdAbv)}{axiom:PropPHdAbv}] \hspace{2.73em}$\PrevHeadHere \dlsubseteq \forall{\roleell_{\nofM{+}1}} \forall{\roler_{\nofM{+}1}} \conceptLvl_{\nofM{+}1} \to \PrevHeadAbove$,
        \item[\desclabel{(PropNoPHdAbv)}{axiom:PropNoPHdAbv}] $\NoPrevHeadHere \dlsubseteq \forall{\roleell_{\nofM{+}1}} \forall{\roler_{\nofM{+}1}} \conceptLvl_{\nofM{+}1} \to \NoPrevHeadAbove$.
    \end{description}

    \item We ensure consistency of the current configuration with the previous transition. Below $(\states, \lettera, \letterb, \states', d) \in \transreldelta$.
    \begin{description}\itemsep0em
        \centering
        \item[\desclabel{(TransiCons)}{axiom:TransiCons}] $\prevtrans{(\states, \lettera, \letterb, \states', d)} \dlsubseteq \PHeadLetter{\letterb} \dland \currstate{\states'} \dland \text{``}\PHeadPos{} + d = \HeadPos{}\text{''}$,
    \end{description}
    where the last right-hand-side expression, specifying decrements or increments of binary encodings of numbers, is implemented in a usual way~\cite[p. 127]{dlbook} via:
    \[
    \bigdlor\limits_{i=1}^{\nofM} \!\Big(\! \conceptA_{i}^0 \dland \conceptB_{i}^1 \dland \!\bigdland\limits_{j=1}^{i-1} (\conceptA_{j}^1 \dland \conceptB_{j}^0 ) \dland
      \!\!\!\!\bigdland\limits_{j=i+1}^{\nofM}\!\!\! \big((\conceptA_{j}^1 \dland \conceptB_{j}^1 ) \dlor (\conceptA_{j}^0 \dland \conceptB_{j}^0 )\big) \!\Big)
    \]
    with $\conceptA := \PHeadPos{}$ and $\conceptB := \HeadPos{}$ if $d = {+}1$, and with $\conceptA$ and $\conceptB$ swapped if $d = {-}1$.
    \item We encode the initial configuration as follows.
      \begin{description}\itemsep0em
        \centering
        \item[\desclabel{(InitConf)}{axiom:InitConf}] $\InitialConf \dlsubseteq \conceptLvl_0 \dland \conceptL \dland \currstate{\initialstate} \dland \bigdland_{i=1}^{\nofM} (\HeadPos{i}^0 \dland \PHeadPos{i}^0) \dland \forall{\roleell_1}\forall{\roler_1} \ldots \forall{\roleell_{\nofM}}\forall{\roler_{\nofM}} (\conceptLvl_{\nofM} \to \conceptLetter{\letterzero})$,
    \end{description}
\end{enumerate}
  
For the knowledge base $\kbKenr$, composed of all the GCIs presented so far, we show correctness in the following lemmas.
Once more, the proofs are routine and similar to the ones from the previous section.

\begin{lemma}\label{prop:enr-conf-trees-satisfies-kbKconf}
  Any enriched configuration tree of $\interE$ is a model of $\kbKenr$.
\end{lemma}
\begin{proof}
Since $\interE$ is a configuration tree by definition, by~\cref{prop:conf-trees-satisfies-kbKconf} we infer $\interE \models \kbKconf$. 
Hence, we may focus on the GCIs presented in this section only.
For the GCIs from the 2nd group, we essentially use the same proof that we used for their ``non-previous'' counterparts the proof of~\cref{prop:conf-trees-satisfies-kbKconf} and thus we do not repeat it here.
Satisfaction of \ref{axiom:PHdAbvCov} and \ref{axiom:PHdAbvDisj} follows by the 4th item of~\cref{def:enr-conf-tree}.
Next, to prove that also \ref{axiom:PropPHdAbv} is satisfied by $\interE$ (the proof for \ref{axiom:PropNoPHdAbv} is analogous) we take any $\addresswordw \in \PrevHeadHere^{\interE}$, which by definition is equal to $\addresswordw_{\textit{phd}}$ and see that the antecedent of the implication on the right hand of \ref{axiom:PropPHdAbv} is satisfied only by $\addresswordw_{\textit{phd}}0$ and $\addresswordw_{\textit{phd}}1$, which are in $\PrevHeadAbove^{\interE}$ by definition.
Next, to show satisfaction of \ref{axiom:TransiCons} we assume $\varepsilon \in \prevtrans{(\states, \lettera, \letterb, \states', d)}^{\interE}$.
Then we have $\varepsilon \in \PHeadLetter{\letterb}^{\interE}$ (by the second to last items of~\cref{def:enr-conf-tree}), $\varepsilon \in \currstate{\states'}^{\interE}$ (by the 2nd item of~\cref{def:enr-conf-tree}) and that $\varepsilon \in \text{``}\PHeadPos{} + d = \HeadPos{}\text{''}^{\interE}$ (by correctness of incrementation/decrementation of binary encodings and by the 1st subitem of the 3rd item of~\cref{def:enr-conf-tree}). 
Thus $\interE \models$ \ref{axiom:TransiCons}.
Finally, $\interE \models$ \ref{axiom:InitConf} follows directly from the last item of~\cref{def:enr-conf-tree}.
\end{proof}

\begin{lemma}\label{lemma:enr-conf-trees-homomorphisms}
  For any model $\interI$ of $\kbKenr$ and any $\domelemd \in \conceptLvl_0^{\interI}$, there is an enriched configuration tree $\interE$ and a homomorphism $\homoh$ from $\interE$ into $\interI$ with $\homoh(\varepsilon) = \domelemd$.
\end{lemma}
\begin{proof}
    We follow the proof scheme of~\cref{lemma:conf-trees-homomorphisms}.
    By~\cref{lemma:conf-trees-homomorphisms}, there is a homomorphism $\homoh$ from a configuration tree $\interC$ to $\interI$ with $\homoh(\varepsilon) = \domelemd$. 
    Moreover, as the symbols outside $\rolesunit \cup \conceptsunit \cup \conceptsconf$ do not appear in~\cref{def:conf-trees} we can assume that $\interC$ interprets them as empty sets.
    Let $\interE = (\DeltaC, \cdot^{\interU})$ be an interpretation that is obtained from changing the meaning of concepts from $\conceptsenrct$ as follows: for any $\conceptC \in \conceptsenrct$ we let $\conceptC^{\interE} := \{ \addresswordw \mid \homoh(\addresswordw) \in \conceptC^{\interC} \}$.
    All other symbols are interpreted as in $\interC$.
    Clearly $\homoh$ is a homomorphism from $\interE$ into $\interI$ with $\homoh(\varepsilon) = \domelemd$ and it suffices to show that $\interE$ is an enriched configuration tree (we already know that it is a configuration tree), which is done by routine investigation of~\cref{def:enr-conf-tree} and the presented GCIs. 

    Existence of the unique concept $\conceptC \in \conceptsptrans$ claimed in the 1st item of~\cref{def:enr-conf-tree} is provided by the first three GCIs, namely the existence is due to \ref{axiom:TrCov} and uniqueness due to \ref{axiom:TrInitDisjt} and \ref{axiom:TrInitDisjttprim}.
    The 2nd item of~\cref{def:enr-conf-tree} is due to $\interI \models$ \ref{axiom:TransiCons} (more precisely, the first conjunct of the rhs).
    Similarly to the proof of~\cref{lemma:conf-trees-homomorphisms}, we establish the existence of a unique $\addresswordw_{\textit{phd}}$ and desired properties of concepts $\PrevHeadHere^{\interE}$, $\NoPrevHeadHere^{\interE}$, $(\PHeadPos{i}^b)^{\interE}$ and $\PHeadLetter{\lettera}^{\interE}$. 
    Since such a proof is nearly identical (modulo adding the ``P'' letter in front of some concept names) to the one from the previous section, we do not repeat the details here.
    Then, the fact that $\addresswordw_{\textit{phd}}$ satisfies $\addresswordw_{\textit{phd}} = \addresswordw_{\textit{head}} + d$ or is equal to $0$ in case $\InitialConf^{\interE} = \{ \varepsilon \}$ is by, respectively, membership of $\varepsilon$ in $\text{``}\PHeadPos{} + d = \HeadPos{}\text{''}^{\interE}$ and $(\forall{\roleell_1}\forall{\roler_1}\ldots\forall{\roleell_{\nofM}}\forall{\roler_{\nofM}} \; (\conceptLvl_\nofM \to \PHeadPos{i}^b))^{\interE}$, guaranteed by $\interI \models$ \ref{axiom:TransiCons} and $\interI \models$ \ref{axiom:InitConf}.
    Finally, the satisfaction of the last item of~\cref{def:enr-conf-tree} is immediate by $\interI \models$ \ref{axiom:InitConf}.
\end{proof}

\section{Describing Accepting Quasi-Runs}\label{sec:quasi-computations}
Recall that a quasi-run $\quasirunR$ of $\atmM$ is simply a tree labelled with configurations of $\atmM$ where the root is labelled with the initial configuration  \(\initialstate \letterzero^{2^{\nofM}}\). Each node representing an existential  configuration has one child labelled with a quasi-successor configuration, while each node representing a universal configuration has two children labelled by quasi-successor configurations obtained via different transitions.

In order to represent an accepting quasi-run by a model, we employ the notion of a \emph{quasi-computation tree} $\inter{Q}$, a structure intuitively defined from some $\quasirunR$ as follows: replace every node of $\quasirunR$ by its corresponding configuration tree, adequately enriched with information about its generating transition and the predecessor configuration. 
The roots of these enriched configuration trees are linked via the $\rolenext$ role to express the quasi-succession relation of $\quasirunR$. 
The roots of enriched configuration trees representing universal configurations are chosen to be labelled with $\conceptL$, their left $\rolenext$-child with $\conceptL$ and their right $\rolenext$-child with $\conceptR$ (both corresponding to existential configurations).
As expected, the $\InitialConf$ concept decorates the root of the distinguished enriched configuration tree that represents $\quasirunR$'s initial configuration. 
As our attention is restricted to \emph{accepting} quasi-runs $\quasirunR$, we require that no enriched configuration tree occurring in $\interQ$ carries a rejecting state. 
We now give a formal definition of such a structure $\interQ$.
\begin{figure}[H]
  \centering
  \vspace{-1em} 
  \hspace{5em}  
  \includegraphics[scale=0.14]{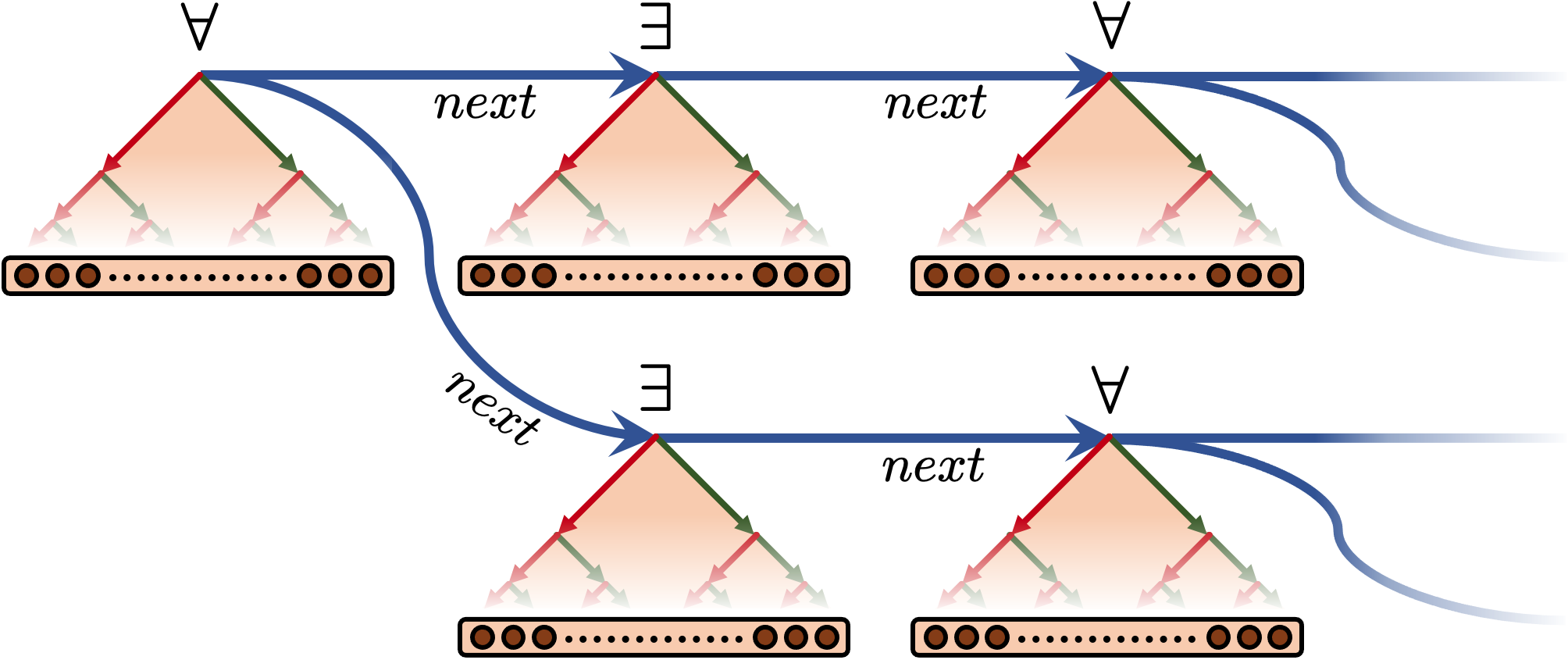}
\end{figure}

\begin{definition}[quasi-computation tree]\label{def:quasi-comp}
  A \emph{quasi-computa\-tion tree} $\interQ$ of $\atmM$ is an interpretation $\interQ = (\DeltaQ, \cdotQ)$ satisfying the following properties:
  \begin{itemize}\itemsep0em
  \item $\DeltaQ := \treeT \times \{0,1\}^{\leq \nofM{+}1}$, where $\treeT$ is\footnote{This is just a scary looking definition of a binary tree in which nodes at the $i$-th level have exactly $2$ children if $i$ is even and exactly one child otherwise.} a prefix-closed subset of $\{\frakone\frakzero, \frakzero\frakzero\}^* \cdot \{\varepsilon,\frakzero,\frakone\}$ with $\frakw\frakone \in \treeT$ implying~$\frakw\frakzero \in \treeT$.
  \item For every $\frakw \in \treeT$, the substructure of $\interQ$ induced by $\{ \frakw \} \times \{ 0,1 \}^{\leq \nofM{+}1}$ is isomorphic to an enriched configuration tree of $\atmM$ via the isomorphism $(\frakw, \addresswordw) \mapsto \addresswordw$.
  \item $(\varepsilon,\frakw) \in \conceptR^{\interQ}$ if $\frakw$ ends with $\frakone$, otherwise $(\varepsilon, \frakw) \in \conceptL^{\interQ}$.
  \item For any $\frakw \neq \frakw'$ and arbitrary $\addresswordw, \addresswordw' \in \{0,1\}^{\leq \nofM{+}1}$ holds $((\frakw, \addresswordw), (\frakw',\addresswordw'))\notin \roles^\interQ$ for any $\roles \in \rolesunit \setminus \set{\rolenext}$.
  \item $\rolenext^{\interQ} \!\setminus\! \{(\domelemd,\domelemd)\mid \DeltaQ{\times}\DeltaQ\} = \{ ((\frakw,\varepsilon),(\frakw\frakb,\varepsilon)) \mid \frakw\frakb, \frakw \in \treeT, \frakb \in \{ \frakzero, \frakone \} \}$.
  \item $\InitialConf^{\interQ} = \{ (\varepsilon,\varepsilon) \}$. 
  \item for any $\frakw\frakzero \in \treeT$ with $(\frakw,\varepsilon) \in \currstates^\interQ$ and $(\frakw,\varepsilon)\in \conceptLetter{\lettera}^\interQ$     
  \begin{itemize}\itemsep0em
  \item if $\frakw\frakone \in \treeT$ then $(\frakw\frakzero,\varepsilon) \in \prevtrans{\transreldelta_1(\states,\lettera)}^\interQ$ and $(\frakw\frakone,\varepsilon) \in \prevtrans{\transreldelta_2(\states,\lettera)}^\interQ$,
  \item if $\frakw\frakone \notin \treeT$ then $(\frakw\frakzero,\varepsilon) \in \prevtrans{\transreldelta_1(\states,\lettera)}^\interQ$ or $(\frakw\frakzero,\varepsilon) \in \prevtrans{\transreldelta_2(\states,\lettera)}^\interQ$.
  \end{itemize}
  \item If $(\frakw, \addresswordw) \in \HeadHere^\interQ$ and $\frakw\frakb \in \treeT$ then $(\frakw\frakb, \addresswordw) \in \PrevHeadHere^\interQ$.
  \item $\currstate{\rejectingstate}^\interQ = \emptyset$ as well as $(\frakw,\varepsilon)\in \currstate{\acceptingstate}^\interQ$ if and only if $\frakw \in \treeT$ and $\frakw\frakzero \not\in \treeT$.
  \end{itemize} 
\end{definition}

We move on to provide an appropriate axiomatisation.
\begin{enumerate}\itemsep0em
  \item We incorporate all axioms from $\kbKenr$ to ensure the indicated substructures correspond to enriched computation trees.
  \item Every non-final existential configuration has one successor configuration while every non-final universal configuration has two. 
  Final configurations do not have any successors.
  Below $\states_{e} \in \statesQforall \setminus \{ \acceptingstate, \rejectingstate\}, \states_{e} \in \statesQexists \setminus \{ \acceptingstate, \rejectingstate\}$, and $\states_{f} \in \{ \acceptingstate, \rejectingstate\}$. 
  \begin{multicols}{2}
  \begin{description}\itemsep0em
      \item[\desclabel{(EConfSucc[$\states_{e}$])}{axiom:EConfSucc}] $\currstate{\states_{e}} \dlsubseteq \exists{\rolenext}.\conceptL \dland \exists{\rolenext}.\conceptR$
  \end{description}
  \columnbreak
  \begin{description}\itemsep0em 
      \item[\desclabel{(AConfSucc[$\states_{a}$])}{axiom:AConfSucc}] $\currstate{\states_{a}} \dlsubseteq \exists{\rolenext}.\topconcept \dland \forall{\rolenext}.\conceptL$
  \end{description}
  \end{multicols}
    \begin{description}\itemsep0em 
      \centering
      \item[\desclabel{(FinConfSucc[$\states_{f}$])}{axiom:FinConfSucc}] $\currstate{\states_{f}} \dlsubseteq \forall{\rolenext}.\botconcept$
  \end{description}
  \item To transfer the previous head position to the successor configurations we employ (for $1 \leq i \leq \nofM, b \in \{0,1\}$):
  \begin{description}\itemsep0em 
      \centering
      \item[\desclabel{(TransHeadPos[$i,b$])}{axiom:TransHeadPosib}] $\conceptLvl_0 \dland \HeadPos{i}^b \dlsubseteq \forall{\rolenext}.\PHeadPos{i}^b$
  \end{description}
  \item For any $\states_{\exists} \in \statesQexists$ we specify that the corresponding configuration tree linked via $\rolenext$-role is a successor configuration of the current one.
    \begin{description}\itemsep0em 
      \centering
      \item[\desclabel{(TransiExistState)}{axiom:TransiExistState}] $\currstate{\states_{\exists}} \dland \HeadLetter{\lettera} \dlsubseteq \bigdlor_{\transitiont \in \transreldelta(\states_{\exists}, \lettera) } \forall{\rolenext}.\prevtrans{\transitiont}$
  \end{description}
  \item For every universal state $\states_{\forall} \in \statesQforall$ and a letter $\lettera$ currently scanned by the head there are only two possible choices of transitions.
  \begin{description}\itemsep0em 
      \centering
      \item[\desclabel{(TransiUnivStateL)}{axiom:TransiUnivStateL}] $\currstate{\states_{\forall}} \dland \HeadLetter{\lettera} \dlsubseteq \forall{\rolenext}.(\conceptL \to \prevtrans{\transreldelta_1(\states_\forall, \lettera)})$
      \item[\desclabel{(TransiUnivStateR)}{axiom:TransiUnivStateR}] $\currstate{\states_{\forall}} \dland \HeadLetter{\lettera} \dlsubseteq \forall{\rolenext}.(\conceptR \to \prevtrans{\transreldelta_2(\states_\forall, \lettera)})$
  \end{description}
  \item Since we want to have accepting quasi-runs of $\atmM$ only, we state that we never encounter the rejecting~state.
    \begin{description}\itemsep0em 
      \centering
      \item[\desclabel{(NoRejectState)}{axiom:NoRejectState}] $\currstate{\rejectingstate} \dlsubseteq \botconcept$
  \end{description}

\end{enumerate}

Let $\tboxT_{\atmM}$ be the set of all GCIs presented so far and let $\aboxA_{\atmM}$ be an ABox composed of a single axiom $\InitialConf(\indva)$ for a fresh individual name $\indva$. 
Put $\kbK_{\atmM} := (\aboxA_{\atmM}, \tboxT_{\atmM})$.
We claim that:

\begin{lemma}\label{prop:quasi-comp-satisfies-kbKATM}
Any accepting quasi-computation tree $\interQ$ of $\atmM$ is a model of $\kbK_{\atmM}$.
\end{lemma}
\begin{proof} 
  To see $\interQ \models \kbKenr$ it suffices to observe that (1) by the 2nd item of~\cref{def:quasi-comp} all the substructures of~$\interQ$ induced by $\{\frakw\} \times \{0,1\}^{\leq \nofM{+}1}$ are isomorphic to some computation tree and hence, by~\cref{prop:enr-conf-trees-satisfies-kbKconf} they satisfy~$\kbKenr$, (2) the use of roles from $\rolesunit \setminus \{ \rolenext \}$ is restricted to enriched configuration trees and hence $\interQ$ satisfies all the GCIs not involving $\rolenext$ and (3) the only GCI involving $\rolenext$ from $\kbKenr$ is~\ref{axiom:leavesnextloop} and it is satisfied in $\interQ$ due to the mentioned isomorphism property.
  Next, satisfaction of \ref{axiom:AConfSucc}, \ref{axiom:TransiUnivStateL} and \ref{axiom:TransiUnivStateR} by $\interQ$ is due to the 7th item (1st subitem) of~\cref{def:quasi-comp}.
  Similarly, we infer that $\interQ \models$ \ref{axiom:EConfSucc} and $\interQ \models$ \ref{axiom:TransiExistState} by the 7th item (2nd subitem) of~\cref{def:quasi-comp}.
  By the last item of~\cref{def:quasi-comp} we immediately conclude $\interQ \models$ \ref{axiom:FinConfSucc} and $\interQ \models$ \ref{axiom:NoRejectState}.
  Hence, it remains to prove satisfaction of \ref{axiom:TransHeadPosib}, which is immediate by the second to last item of~\cref{def:quasi-comp}.
\end{proof}

\begin{lemma}\label{lemma:quasi-computations-homomorphisms}
For any model $\interI$ of $\kbK_{\atmM}$ there exists an accepting quasi-compu\-ta\-tion tree $\interQ$ and a homomorphism $\homoh : \interQ \to \interI$ with $\homoh(\varepsilon, \varepsilon) = \indva^{\interI}$.
\end{lemma}
\begin{proof}
  We construct a tree $\treeT$ and its origin function $\homof: \treeT \to \interI$ as follows.
  First, let $\varepsilon \in \treeT$ and $\homof(\varepsilon) = \indva^{\interI}$.
  We next proceed as follows: take any word $\frakw \in \treeT$ and consider three cases:
  \begin{itemize}\itemsep0em
    \item $\homof(\frakw)$ is labelled with a non-final universal state. Hence, by the first axiom provided, we know that $\homof(\frakw)$ has at least two $\rolenext$-successors, one of which is in $\conceptL^{\interI}$ and the other in $\conceptR^{\interI}$. Call them, respectively, $\domeleme_l, \domeleme_r$. 
    Hence, we extend $\treeT$ with the words $, \frakw\frakone$ and extend $\homof$ with $\homof(\frakw\frakzero) = \domeleme_l$ and $\homof(\frakw\frakone) = \domeleme_r$. Repeat the process from $\frakw\frakzero$ and $\frakw\frakone$.
    \item $\homof(\frakw)$ is labelled with a non-final existential state. Then we take its $\rolenext$-successor $\domeleme$ and extend $\treeT$ with $\frakw\frakzero$ and $\homof$ with $\homof(\frakw\frakzero) = \domeleme$. Repeat the process from $\frakw\frakzero$.
    \item $\homof(\frakw)$ is labelled with a final state. No action required.
  \end{itemize}
  We associate a word $\frakw \in \treeT$ with an enriched configuration tree $\interE_{\frakw}$ such that there is a homomorphism $\homog_{\frakw}$ from $\interE_{\frakw}$ to $\interI$ with $\homof(\frakw) = \homog_{\frakw}(\varepsilon)$. The existence of $\interE_{\frakw}$ and $\homog_{\frakw}$ is provided by~\cref{lemma:enr-conf-trees-homomorphisms}.
  Finally, we decorate each node of $\interE_{\frakw}$ with ``Pr'' concepts as suggested by the homomorphism $\homog_{\frakw}$.
  A $\treeT$-quasi-computation tree $\interQ$ is then defined by stipulating  
  that, for every $\frakw \in \treeT$, the substructure of $\interQ$ induced by $\{\frakw\} \times \{0,1\}^{\leq \nofM{+}1}$
  be isomorphic to the decorated $\interE_{\frakw}$. 
  The homomorphism $\homoh: \DeltaQ \to \interI$ is then defined componentwise by $(\frakw, \addresswordw) \mapsto \homog_{\frakw}(\addresswordw)$, essentially taking the disjoint unions of the homomorphisms for all enriched configuration trees. Since all the roles except $\rolenext$ are restricted to the components and we made sure that the roots of $\interQ$ were created from the elements linked via $\rolenext$-roles, we conclude that $\homoh$ is the claimed homomorphism.
\end{proof}

\section{Detecting Faulty Runs with a Single CQ}\label{sec:query}
We finally have reached the point where querying comes into play. 
Our last goal is to design \emph{one single} conjunctive query that detects ``faulty configuration progressions'' in quasi-computation trees, meaning that it matches a pair of two positions in consecutive configuration trees representing the same cell and being untouched by the head of $\atmM$ yet storing different letters.
Note that the lack of such cells in a quasi-computation tree means that any two consecutive configuration trees represent not only quasi-successor configuration but actually proper successors and hence the structure as such even represents a ``proper'' run. We start by formalising our requirements to such a query:
\begin{lemma}\label{lemma:existence-of-cq-spoiling-fun}
There exists a CQ $\queryq_{\atmM}$ of size polynomial in $\nofM$ with two distinguished variables $\varx, \vary$ such that for all quasi-computation trees $\interQ$ we have $\interQ \modelsmatch{\matchpi} \queryq_{\atmM}$ iff 
there exists a word $\frakw$, a letter $\frakb$ and a word $\addresswordw$ of length~$\nofM{+}1$~s.t.: 
\begin{itemize}\itemsep0em
  \item $\matchpi(\varx) = (\frakw, \addresswordw)$, $\matchpi(\vary) = (\frakw\frakb, \addresswordw)$,
  \item $\matchpi(\vary) \in \NoPrevHeadAbove^{\interQ}$,
  \item $\matchpi(\varx) \in \conceptzero^{\interQ}$ and  $\matchpi(\vary) \in \conceptone^{\interQ}$.
\end{itemize}
\end{lemma}
Note the asymmetry in the 3rd bullet point above -- we ignore the reverse constellation. Yet, due to our encoding if the reverse situation occurs then so does the original one. Hence, every mismatch in a sense causes two inconsistencies from the point of $\nofM{+}1$-level nodes. This solves the mystery of introducing level $\nofM{+}1$ in our configuration trees and the particular encoding of tape symbols: it is crucial for catching faulty progressions by using one single CQ.
Before proving~\cref{lemma:existence-of-cq-spoiling-fun} we show how it implies the main theorem of our paper, namely:
\begin{thm}\label{thm:main}
Conjunctive query entailment over $\ALCself$ knowledge bases is $\TwoExpTime$-hard.
\end{thm}
\begin{proof}
  Since $\coTwoExpTime {=} \TwoExpTime$, it is sufficient to show that CQ non-entailment over $\ALCself$ KBs is $\TwoExpTime$-hard. 
  Take $\kbK_{\atmM}$ as defined in~\cref{sec:quasi-computations} and $\queryq_{\atmM}$ as given by~\cref{lemma:existence-of-cq-spoiling-fun}.
  We claim that $\kbK_{\atmM} \not\models \queryq_{\atmM}$ iff~$\atmM$ is accepting.
  The ``if'' direction is easy: we take an accepting run of $\atmM$ and turn it into quasi-computation tree $\interQ$.
  By~\cref{prop:quasi-comp-satisfies-kbKATM} we conclude that $\interQ \models \kbK_{\atmM}$. 
  We also have that $\interQ \not\models \queryq$ due to the fact that any two consecutive configuration trees represent proper successor configurations.
  For the second direction it suffices to show that if $\atmM$ is not accepting then $\kbK_{\atmM} \models \queryq_{\atmM}$.
  Indeed, assume that $\atmM$ is not accepting and consider an arbitrary model $\interI$ of $\kbK_{\atmM}$ (in case $\kbK_{\atmM}$ is unsatisfiable then trivially $\kbK_{\atmM} \models \queryq_{\atmM}$).
  By~\cref{lemma:quasi-computations-homomorphisms} there is a quasi-computation tree $\interQ$ and a homomorphism $\homoh : \interQ \to \interI$ with $\homoh(\varepsilon, \varepsilon) = \indva^{\interI}$. 
  But this quasi-computation tree must represent a ``faulty'' run -- in the opposite case it would correspond to an accepting run of $\atmM$, which does not exist by assumption.
  Hence there must ba a match of $\queryq_{\atmM}$ to $\interQ$. Yet, as query matches are preserved under homomorphisms, we conclude $\interI \models \queryq_{\atmM}$. Thus all models $\interI$ of $\kbK_{\atmM}$ have matches of $\queryq_{\atmM}$, which implies~$\kbK_{\atmM} \models \queryq_{\atmM}$.
\end{proof}

In the forthcoming query definitions, we employ a convenient naming scheme. 
By writing $\queryq[\varx,\vary]$ we indicate that the variables $\varx, \vary \in \queryVarq$ are \emph{global} (\ie the same across (sub)queries that we might join together) while its other variables are \emph{local} (\ie pairwise different from any variables occurring in other queries –- this can always be enforced by renaming).
Going back to the query, we proceed as follows. 
We first prepare a query $\qmain[\varx, \vary]$ with two global distinguished variables $\varx,\vary$ that relates any two domain elements whenever they are leaf nodes of consecutive computation trees.
Then $\qmain[\varx,\vary]$ is combined with queries $\qaddr{i}[\varx,\vary]$ for all $1 \leq i \leq \nofM{+}1$ with the intended meaning that $\varx$ and $\vary$ have the same $i$-th bit of their addresses.
Additionally, our final query will require that $\varx$ be mapped to a node satisfying $\conceptzero$ and $\vary$ to a node satisfying $\conceptone$ and $\NoPrevHeadHere$.

To construct $\qmain[\varx,\vary]$ we essentially employ~\cref{lemma:from-root-to-leaf-path-via-loops}.
\begin{lemma}\label{lemma:first-query}
There exists a CQ $\qmain[\varx,\vary]$ such that for any quasi-computation tree $\interQ$ we have that $M_{\qmain} := \{ (\matchpi(\varx), \matchpi(\vary)) \mid \interQ \modelsmatch{\matchpi} \qmain \}$
is equal to any pair of leaves of two consecutive configuration trees of $\interQ$. 
Formally:
\[ M_{\qmain} = \big\{\! \left( (\frakw, w),\! (\frakw\frakb, v) \right) \in \DeltaQ \mid |w|{\,=\,}|v|{\,=\,}\nofM{+}1, \frakb{\,\in\,} \{0,1\}\! \big\}. \]
\begin{figure}[H]
\centering
\includegraphics[scale=0.2]{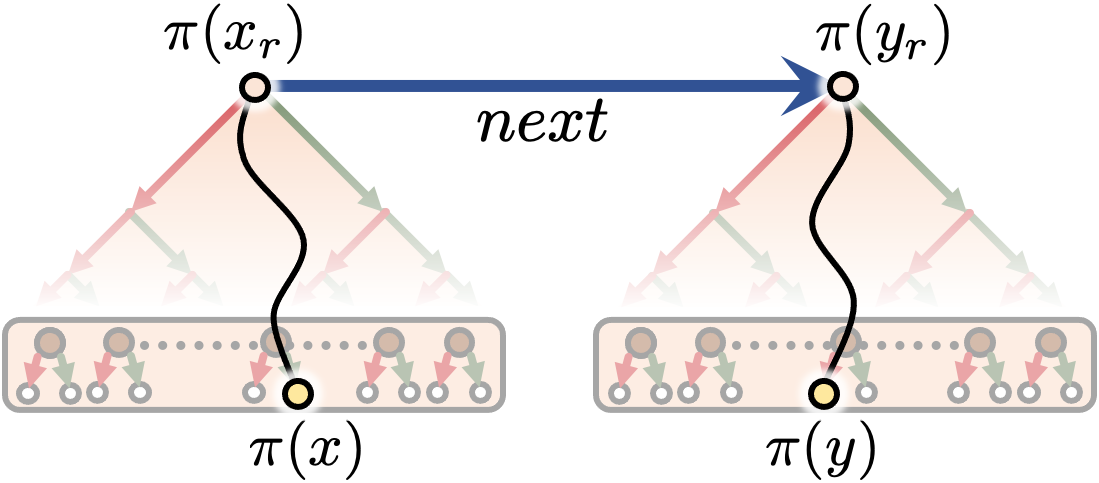}
\end{figure}
\end{lemma}
\begin{proof} 
  It suffices to take $\qmain := \qrl[\varx_r, \varx] \land \rolenext(\varx_r, \vary_r) \land \qrl[\vary_r, \vary]$.
  Let $\interQ \modelsmatch{\matchpi} \qmain$. 
  That $M_{\qmain}$ is a superset of the set above follows from the fact that quasi-computation trees are computation units and hence, containedness follows by~\cref{lemma:query-from-root-to-leaf}. We now focus on the other direction.
  Note that by the 5th item of~\cref{def:quasi-comp} we know that $\matchpi(\varx_r)$ and $\matchpi(\vary_r)$ must be two distinct roots of enriched configuration trees $\interE_{\varx_r}, \interE_{\vary_r}$. 
  By the 4th item of~\cref{def:quasi-comp} we know that the interpretation of the $\roler$s and $\roleell$s is restricted to pairs of domain elements located inside the same enriched configuration tree (and by their definition to configuration trees and by their definition to configuration units).
  Since $\qrl$ only employs the roles $\roleell_i, \roler_i$ and the concepts $\conceptLvl_0, \conceptLvl_{\nofM{+}1}$ we conclude that $\qrl$ has exactly the same set of matches in $\interE_{\varx_r}$ as in its underlying unit.
  Hence, by~\cref{lemma:query-from-root-to-leaf} we know that $\varx$ (resp. $\vary$) is indeed mapped to a leaf of $\interE_{\varx_r}$ (resp. to a leaf of $\interE_{\vary_r}$), which finishes the proof.
\end{proof}

The next part of our query construction focuses on sub-queries $\qaddr{i}[\varx,\vary]$ that are meant to relate leaves having equal $i$-th bits of addresses. 
In order to construct it we combine together several smaller queries, written in path syntax below.
\begin{itemize}\itemsep0em
  \item We let $\qtopdown[\varx, \vary] := (\roleell_{1}; \roler_{1}; \ldots; \roleell_{\nofM{+}1}; \roler_{\nofM{+}1})(\varx, \vary)$ define the \emph{top-down query}. 
  It intuitively traverses an enriched configuration tree in a top-down manner. 
  Note that $\qtopdown[\varx,\vary]$ is actually the major sub-query of $\qrl[\varx,\vary]$. 
  \item The \emph{$\roleell_i$-top-down query} $\qLtopdowni{i}[\varx,\vary]$ is similar to $\qtopdown[\varx, \vary]$, but with the $\roleell_i;\roler_i$ part replaced by just $\roleell_i$. 
  The intended behavior is that again a tree is traversed from root to leaves, but this time, an $\roleell_i$ edge must be crossed when going from the $(i-1)$-th to the $i$-th level.
  The \emph{$\roler_i$-top-down query} $\qRtopdowni{i}[\varx,\vary]$ is defined alike, by replacing $\roleell_i;\roler_i$ in $\qtopdown[\varx, \vary]$ with $\roler_i$.
\end{itemize}

An important ingredient in the construction is the query $\qithbit{i}{0}[\varx, \vary]$ defined as follows:
\[
  \conceptLvl_{\nofM{+}1}(\varx) \land \qLtopdowni{i}[\varx'\!,\! \varx] \land \rolenext(\varx'\!,\! \vary') \land \qLtopdowni{i}[\vary'\!,\! \vary] \land \conceptLvl_{\nofM{+}1}(\vary).
\]
In total analogy, we define $\qithbit{i}{1}[\varx, \vary]$ by using $\qRtopdowni{i}$ instead of $\qLtopdowni{i}$. 
Any match $\matchpi$ of the query $\qithbit{i}{b}[\varx, \vary]$ instantiates  the variables $\varx$ and $\vary$ in a quasi-computation tree $\interQ$ according to one of the following two scenarios: either $\matchpi(\varx) = \matchpi(\vary)$ or $\matchpi(\varx)$ and $\matchpi(\vary)$ are leaves in two consecutive enriched configuration trees inside the quasi-computation tree and both of these leaves have their $i$-th address bit set to $b$.
The above intuition meets its formalisation in the next lemma:
\begin{lemma}\label{lem:bitpreserve}
Let $\interQ$ be a quasi-computation tree and let $M_{\qithbit{i}{b}}\! = \!\{(\matchpi(\varx),\! \matchpi(\vary)) \mid \interQ \modelsmatch{\matchpi} \qithbit{i}{b} \}$  for $b {\,\in\,} \{0,\!1\}$. 
Then $M_{{\qithbit{i}{b}}}$ is equal to the union of $M^b_1 := \{ ((\frakw, \addresswordw), (\frakw, \addresswordw)) \}$
and $M^b_2 := \{ ((\frakw, \addresswordu b \addresswordv), (\frakw\frakb, \addresswordu' b \addresswordv')) \mid |\addresswordu| {=} |\addresswordu'| {=} i{-}1 \}$.
\end{lemma}
\begin{proof}
  We show the statement for $b=0$, the case for $b=1$ then follows by symmetry.
  First we show~$M^0_1 \subseteq M_{\qithbit{i}{0}}$. This is easy: for any leaf $\domelemd = (\frakw, \addresswordw)$ we map all variables of $\qithbit{i}{0}[\varx,\vary]$ into $\domelemd$; this is a match due to the presence of all the self-loops at the leaves.
  To show $M^0_2 \subseteq M_{\qithbit{i}{0}}$ we take any $\domelemd = (\frakw, \addresswordw)$ and~$\domeleme = (\frakw\frakb, \addresswordv)$. 
  Let $\matchpi$ be a variable assignment that maps $\varx$ to $\domelemd$, $\vary$ to $\domeleme$, $\varx'$ to $(\frakw, \varepsilon)$, $\vary'$ to $(\frakw\frakb, \varepsilon)$.
  The variables of $\qLtopdowni{i}[\varx', \varx]$ are mapped to $(\frakw, \addresswordw_j)$, where $\addresswordw_j$ is the prefix of $\addresswordw$ of length $j$ following the path from $(\frakw, \varepsilon)$ to $(\frakw, \addresswordw)$ level-by-level.
  We stress that $( (\frakw, \addresswordw_{i{-}1}), (\frakw, \addresswordw_i)) \in \roleell_{i}^\interQ$ holds, which is crucial for the construction to work and that every $(\frakw, \addresswordw_j)$ node has all $\roleell$- and $\roler$-loops.
  The variables of $\qLtopdowni{i}[\vary', \vary]$ are mapped analogously.
  After noticing that $\domelemd, \domeleme \in \conceptLvl_{\nofM{+}1}^{\interQ}$ and that $(\matchpi(\varx'), \matchpi(\vary')) \in \rolenext^{\interQ}$ holds, we conclude that $\matchpi$ is clearly a match of $\qithbit{i}{0}[\varx,\vary]$ to $\interQ$.

  Now we focus on showing that $M_{\qithbit{i}{0}[\varx,\vary]} \subseteq M^0_1 \cup M^0_2$.
  Take any match $\matchpi$ and note that $\varx, \vary$ must be mapped to leaves. 
  For $\matchpi(\varx')$ and $\matchpi(\vary')$ we consider the two cases:
  \begin{enumerate}\itemsep0em
    \item $\matchpi(\varx') = \matchpi(\vary')$. 
    As the roots do not have $\rolenext$-loops, $\matchpi(\varx')$ must be a leaf. 
    This implies that all variables of $\qLtopdowni{i}[\varx', \varx]$ map into a single domain element (otherwise we would not reach a leaf after traversing such path). 
    Arguing similarly we infer that all variables of $\qLtopdowni{i}[\vary', \vary]$ are mapped to the same element. 
    Thus $\matchpi(\varx) = \matchpi(\vary)$ holds.

    \item $\matchpi(\varx') \neq \matchpi(\vary')$. 
    Since all incoming $\rolenext$ roles from leaves are self-loops, we conclude that $\matchpi(\varx')$ is the root of some enriched quasi-computation tree and $\matchpi(\vary')$ is the root of some corresponding quasi-successor in $\interQ$ (by the definition of $\rolenext^{\interQ}$). 
    By the satisfaction of $\qLtopdowni{i}[\varx', \varx]$ we know that there exists a sequence of domain elements contributing to a path from $\matchpi(\varx')$ to $\matchpi(\varx)$ witnessing its satisfaction. 
    Moreover, note that since the subquery $\qLtopdowni{i}[\varx', \varx]$ leads from the root to a leaf it implies that we necessarily cross the $\roleell_i$ role at the $i{-}1$-th level, meaning that the $i$-th bit of the address of $\matchpi(\varx)$ is equal to $0$.
    Thus we infer that $\matchpi(\varx) \in (\conceptadr{i}^0)^{\interQ}$.
    Reasoning analogously we conclude that $\matchpi(\vary) \in (\conceptadr{i}^0)^{\interQ}$, which finishes the proof.\qedhere
  \end{enumerate}
\end{proof}

We are now ready to present the query $\qaddr{i}[\varx,\vary]$ pairing leaves in consecutive enriched configuration trees with coinciding $i$-th address bit:
\[\qaddr{i}[\varx,\vary] := \qmain[x,y] \land \qithbit{i}{0}[\varx, \varz] \land \qithbit{i}{1}[\varz, \vary]\]
\begin{lemma}\label{lemma:second-query}
For any quasi-computation tree $\interQ$ we have that $M_{\qaddr{i}} = \{ (\matchpi(\varx), \matchpi(\vary)) \mid \interQ \modelsmatch{\matchpi} \qaddr{i}[\varx,\vary] \}$ is equal to the leaf pairs in two consecutive enriched configuration trees of $\interQ$ having equal $i$-th bit of address, formally:
\[
M_{\!\qaddr{i}} = M_{\qmain} \cap \Big( \big( {\conceptadr{i}^0}^{\interQ} {\times} {\conceptadr{i}^0}^{\interQ} \big) \cup \big( {\conceptadr{i}^1}^{\interQ} {\times} {\conceptadr{i}^1}^{\interQ} \big)\Big)
\]
\end{lemma}
\begin{proof}
By employing the definition of the query, \cref{lem:bitpreserve} and relational calculus we conclude that $M_{\!\qaddr{i}} = M_{\qmain} \cap \big(M_{\qithbit{i}{0}} \circ M_{\qithbit{i}{1}}\big) = M_{\qmain} \cap \big((M^0_1 \cup M^0_2) \circ (M^1_1 \cup M^1_2)\big) = M_{\qmain} \cap \big(M^0_1 \cup M^1_2 \cup M^0_2\big) =M^1_2 \cup M^0_2$ \qedhere
\end{proof}

We are finally ready to present our query 
\[
  \queryq_{\atmM} := \bigwedge_{i=1}^{\nofM{+}1} \qaddr{i}[\varx,\vary] \land \NoPrevHeadAbove(\vary) \land \conceptzero(\varx) \land \conceptone(\vary)
\]
by means of which we can conclude with the proof of~\cref{lemma:existence-of-cq-spoiling-fun}.
\begin{proof}[Proof of~\cref{lemma:existence-of-cq-spoiling-fun}]
Let $\queryq_{\atmM}$ as defined above and observe that its size is clearly polynomial in $\nofM$. 
Note that $\queryq_{\atmM}$ satisfies our requirements: 
The 1st item follows from two lemmas: the fact that $\varx$ and $\vary$ are mapped to leaves of two consecutive enriched configuration trees follows from~\cref{lemma:first-query} and the fact that $\varx$ and $\vary$ are mapped to nodes having equal addresses follows from~\cref{lemma:second-query} applied for every $1 \leq i \leq \nofM{+}1$.
The 2nd and the 3rd points hold since we supplemented our query with $\NoPrevHeadAbove(\vary) \land \conceptzero(\varx) \land \conceptone(\vary)$.
\end{proof}

\section{Conclusions}\label{sec:conc}
Conjunctive query entailment for $\ALCself$ is, in fact $\TwoExpTime$-complete, where membership follows from much stronger logics \cite[Thm~4.3]{CalvaneseEO09}.
Hardness, shown in this paper, came as a quite surprise to us (in fact, we spent quite some time trying to prove \mbox{\sc{ExpTime}}-membership). The key insight of our proof (and maybe the take-home message from this paper) is that the presence of $\Self$ allows us to mimic case distinction over paths (and hence the handling of disjunctive information) through concatenation, by providing the opportunity for one of the two disjuncts to idle by ``circling in place''.

On a last note, our result also holds for plain $\ALCself$ TBoxes, since the only ABox assertion $\InitialConf(\indva)$ can be replaced by the GCI $\topconcept \dlsubseteq \exists{\role{aux}}.\InitialConf$ for an auxiliary role name $\role{aux}$.   

\section*{Acknowledgements}
This work was supported by the ERC through the Consolidator Grant No.~771779 (\href{https://iccl.inf.tu-dresden.de/web/DeciGUT/en}{DeciGUT}).
\begin{figure}[h]
    \centering
    \includegraphics[scale=0.05]{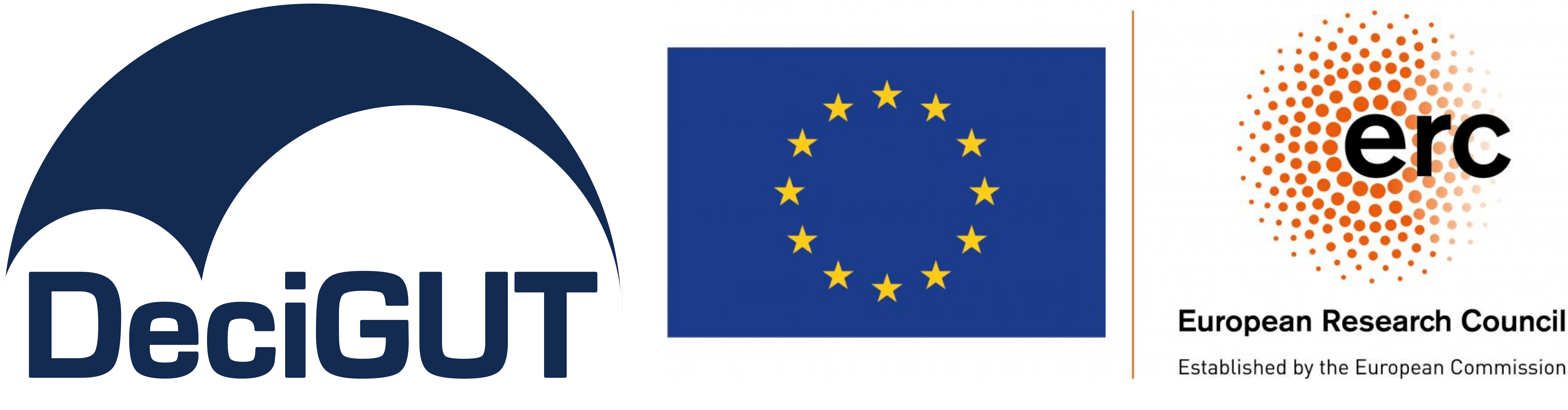}
\end{figure}

\pagebreak
\bibliographystyle{alpha}
\bibliography{references}
\end{document}